\documentclass{article}

\usepackage[nonatbib,final]{neurips_2020}
\usepackage{microtype}
\usepackage{graphicx}
\usepackage{booktabs}
\usepackage{float}
\usepackage{tikz}

\usetikzlibrary{shapes,decorations,arrows,calc,arrows.meta,fit,positioning}
\tikzset{
    -Latex,auto,node distance =1 cm and 1 cm,semithick,
    state/.style ={ellipse, draw, minimum width = 0.7 cm},
    point/.style = {circle, draw, inner sep=0.04cm,fill,node contents={}},
    bidirected/.style={Latex-Latex,dashed},
    el/.style = {inner sep=2pt, align=left, sloped}
}

\usepackage{subcaption}

\usepackage{enumitem}
\usepackage{nicefrac}
\usepackage{tikz}
\usepackage{wrapfig}
\usepackage{algorithmic}
\usetikzlibrary{bayesnet}
\usepackage[numbers]{natbib}
\usepackage{cool}
\usepackage{eulervm}

\usepackage{amsmath}
\usepackage[scaled=0.92]{PTSans}
\usepackage{amsthm}
\usepackage[english]{babel}
\usepackage[parfill]{parskip}
\usepackage{afterpage}
\usepackage{framed}

\usepackage{graphicx}
\usepackage[labelfont=bf]{caption}
\usepackage{booktabs}

\usepackage[algoruled,linesnumbered]{algorithm2e}
\usepackage{listings}
\usepackage{fancyvrb}
\fvset{fontsize=\normalsize}

\usepackage{natbib}

\usepackage[colorlinks=true,linktoc=all]{hyperref}
\usepackage[all]{hypcap}
\hypersetup{citecolor=blue}
\hypersetup{linkcolor=cyan}
\usepackage[nameinlink]{cleveref}

\usepackage
[acronym,smallcaps,nowarn,section,nonumberlist]{glossaries}
\glsdisablehyper{}

\lstdefinestyle{mystyle}{
    commentstyle=\color{OliveGreen},
    keywordstyle=\color{BurntOrange},
    numberstyle=\tiny\color{black!60},
    stringstyle=\color{MidnightBlue},
    basicstyle=\ttfamily,
    breakatwhitespace=false,
    breaklines=true,
    captionpos=b,
    keepspaces=true,
    numbers=left,
    numbersep=5pt,
    showspaces=false,
    showstringspaces=false,
    showtabs=false,
    tabsize=2
}
\usepackage{xspace}

\DeclareRobustCommand{\mb}[1]{\mathbold{#1}}

\DeclareMathOperator*{\argmin}{arg\,min}

\renewcommand{\mid}{~\vert~}

\newcommand{\vz}{{\mathtt{z}}}

\newcommand{\vt}{{\tilde{\mathtt{t}}}}
\newcommand{\vts}{\mathtt{t}^*}
\newcommand{\vtt}{\mathtt{t}^\prime}
\newcommand{\vthat}{\hat{\mathtt{t}}}

\newcommand{\mba}{\mb{a}}

\newcommand{\mbt}{\mb{t}}

\newcommand{\mby}{\mb{y}}
\newcommand{\mbz}{\mb{z}}

\newcommand{\mbeta}{\mb{\eta}}

\newcommand{\supp}{\textrm{supp}}
\newcommand{\sign}{\textrm{sign}}

\renewcommand{\E}{{\mathbb{E}}}

\newcommand{\cG}{\mathcal{G}}

\newcommand{\cN}{\mathcal{N}}

\newcommand{\g}{\,\vert\,}

\newcommand{\indep}{\rotatebox[origin=c]{90}{$\models$}}

\newtheorem*{assumption*}{Assumption}
\newtheorem{theorem}{Theorem}
\newtheorem*{theorem*}{Theorem}

\newcommand{\fullder}[2]{\frac{d #1}{d #2}}

\newacronym{ATE}{ate}{average treatment effect}
\newacronym[longplural = deep exponential families]{DEF}{def}{deep exponential family}

\newacronym{gwas}{gwas}{genome-wide association studies}
\newacronym{lmm}{lmm}{linear mixed model}
\newacronym{pca}{pca}{principle component analysis}

\newacronym{mci}{mte}{effects from multiple treatments}

\newacronym{vde}{vde}{variational decoupling}
\newacronym{out}{out}{outcome regression}

\newacronym{ours}{lode}{Level-set Orthogonal Descent Estimation}
\newacronym{efc}{efc}{estimation with functional confounders}
\newacronym{c-red}{c-redundancy}{causal redundancy}
\newacronym{alt-red}{mono-redundancy}{monotonic redundancy}
\newacronym{rmse}{rmse}{root-mean-squared error}
\newacronym{func-pos}{f-positivity}{functional positivity}
\newacronym{eff-pos}{l-positivity}{level set positivity}
\newacronym{noisy-efc}{noisy-efc}{estimation with noisy functional confounders}
\newacronym{ci}{obs-ci}{observational causal inference}

\newcommand{\tci}{traditional \gls{ci}\xspace}
\newcommand{\Tci}{Traditional \gls{ci}\xspace}

\newcommand{\funcpos}{\gls{func-pos}}

\newcommand{\funcints}{functional interventions}

\newcommand{\wrt}{with respect to }

\newcommand{\condeff}[1][]{
\ifthenelse{\equal{#1}{}}{conditional effect}{conditional effect#1}
}
\newcommand{\avgeff}{average effect}

\newacronym{snp}{snp}{single nucleotide polymorphism}
\newcommand{\snps}{\glspl{snp}}

\newlength\titlepageskip

\title{Causal Estimation with Functional Confounders}

\author{
Aahlad Puli\textsuperscript{1} \hspace{50pt} Adler J. Perotte\textsuperscript{2} \hspace{50pt} Rajesh Ranganath\textsuperscript{1,3}  \\
\texttt{aahlad@nyu.edu} \hspace{10pt} \texttt{adler.perotte@columbia.edu} \hspace{10pt} \texttt{rajeshr@cims.nyu.edu}\\\\
\textsuperscript{1}Computer Science, New York University, New York, NY 10011 \\
\textsuperscript{2}Biomedical Informatics, Columbia University, New York, NY 10032\\
\textsuperscript{3}Center for Data Science, New York University, New York, NY 10011\\
}

\begin{document}

\maketitle
\begin{abstract}
  Causal inference relies on two fundamental assumptions: \emph{ignorability} and \emph{positivity}.
  We study causal inference when the true confounder value can be expressed as a function of the observed data; we call this setting \emph{\gls{efc}}.
  In this setting ignorability is satisfied, however positivity is violated, and causal inference is impossible in general.
  We consider two scenarios where causal effects are estimable.
  First, we discuss interventions on a part of the treatment called \emph{\funcints} and a sufficient condition for effect estimation of these interventions called \emph{\acrlong{func-pos}}.
  Second, we develop conditions for nonparametric effect estimation based on the gradient fields of the functional confounder and the true outcome function.
  To estimate effects under these conditions, we develop \gls{ours}.
  Further, we prove error bounds on \gls{ours}'s effect estimates, evaluate our methods on simulated and real data, and empirically demonstrate the value of \gls{efc}.
\end{abstract}

\section{Introduction}

Determining the effect of interventions on outcomes using observational data lies at the core of many fields like medicine, economic policy, and genomics.
For example, policy makers estimate effects to elect whether to invest in education or job training programs.
In medicine, doctors use effects to design optimal treatment strategies for patients.
Geneticists perform \gls{gwas} to relate genotypes and phenotypes.
In observational data, there could exist unobserved variables that affect both the intervention and the outcome, called confounders.
A necessary condition for the causal effect to be identified is that all confounders are observed; called \emph{ignorability}.
If ignorability holds, a sufficient condition for causal effect estimation is adequate variation in the intervention after conditioning on the confounders; this condition is called \emph{positivity}.

The data apriori does not differentiate between confounders and interventions.
It is the practitioners that select interventions of interest from all pre-outcome variables (variables that occur before the outcome).
Then, assuming knowledge of the data generating mechanism, practitioners can label certain variables amongst the remaining pre-outcome variables as confounders.
This corresponds to indexing into the set of pre-outcome variables.

In certain problems the confounders are specified as 
a function of the pre-outcome variables that does not simply index
into the set of pre-outcome variables. For a concrete example, consider
\gls{gwas}. The goal in \gls{gwas} is to estimate the influence of
genetic variations on phenotypes like disease risk.
In \gls{gwas},
population and family structures both result in certain genetic variations and affect phenotypes and therefore, are confounders~\cite{astle2009population}.
Practitioners specify these confounders by using the genetic similarity between individuals~\cite{lippert2011fast,price2006principal,yu2006unified}, which is
a function of the genetic variations.
When the confounders are a function of the same pre-outcome variables that define the interventions, positivity is violated.
Then, the class of interventions whose effects are estimable is not well-defined.

We study causal effect estimation in such settings, where a function of the
pre-outcome variables provides the confounder and these same pre-outcome variables
define the intervention. We call this \glsreset{efc}\gls{efc}.
In \gls{efc}, one column in the observed data is the outcome and all others are  pre-outcome variables. 
We assume access to a function $h(\cdot)$ that takes as input the pre-outcome variables
and returns the value of the confounder.
Further, we assume these confounders give us ignorability.
In settings like \gls{gwas}, the function $h$ reflects the practitioner-specified function that captures the genetic variation influenced by the population structure.
In \tci, $h(\cdot)$ reflects the selection of certain variables in the data and labelling them as confounders.
In \gls{efc}, two different values of the confounder are never observed for the same setting of the pre-outcome variables.
This means that positivity is violated and the effects of only certain interventions may be estimable.

We address this issue in two ways.
First, we investigate a class of plausible interventions that are \emph{functions} of the observed pre-outcome variables, called \funcints.
We develop a sufficient condition to estimate the effects of said \funcints, called \glsreset{func-pos}\funcpos.
Second, we consider intervening on all pre-outcome variables, called the \emph{full} intervention.
We develop a sufficient condition to estimate the effect of the \emph{full} intervention, called \gls{c-red}.
For an intervention, given a confounder value, \gls{c-red} allows us to compute a surrogate intervention such that the conditional effect of the surrogate is equal to that of the original intervention.
We also show that such surrogate interventions exist only under a certain condition that we call Effect Connectivity, that is necessary for nonparametric effect estimation in \gls{efc}.
This condition is satisfied by default in \tci if ignorability and positivity hold.
Then, we develop an algorithm for causal estimation assuming \gls{c-red}, called \glsreset{ours}\gls{ours}, which estimates effects using surrogate interventions.
If the surrogate is not estimated well, \gls{ours}'s estimates are biased.
We establish bounds on this bias that capture the mitigating effect of the smoothness of the true outcome function.

\paragraph{Related work}

The problem of \glsreset{gwas}\gls{gwas} is to estimate the effect of genetic variations(also called \glspl{snp}) on the phenotype~\cite{visscher201710}.
The ancestry of the subjects acts as a confounder in \gls{gwas}.
In \gls{gwas} practice, \gls{pca} and \glspl{lmm} are used to compute this confounding structure~\cite{price2006principal,yu2006unified}.
\citet{lippert2011fast} suggest estimating the confounders and effects on \emph{separate} subsets of the \glspl{snp}.
This separation disregards the confounding that is captured in the interaction of the two subsets of \glspl{snp}.
\Gls{gwas} is a special case of \glsreset{mci}\gls{mci} where the confounder value
is specified via optimization as a function of the pre-outcome variables~\cite{ranganath2018multiple,wang2018blessings}.
In all these settings, positivity is violated and not all effects are estimable.
We provide an avenue for nonparametric effect-estimation of the full intervention under a new condition, \gls{c-red}.

\paragraph{\glsreset{ci}\Tci review}\label{sec:obs-ci}
We setup causal inference with Structural Causal Models \citep{pearl2009causal} and use $do(\mbt=\vts)$ to denote making an intervention.
Let $\mbt$ be a vector of the interventions, $\mbz$ be the confounder, and $\mby$ be the outcome.
Let $\mbeta \sim p(\mbeta) (\mbeta \indep (\mbz, \mbt))$ be noise.
With $f$ as the \emph{outcome function}, we define the causal model for \tci as
\footnote{We focus on $f$ that generates $\mby$ from $\mbt, \mbz$. SCMs generally specify the function that generates $\mbt$ from $\mbz$ also.}:
\begin{align*}
  \mbz \sim p(\mbz), \quad \mbt \sim p(\mbt \g \mbz), \quad y = f(\mbt, \mbz, \mbeta).
\end{align*}
Let $p(y,\mbz,\mbt)$ denote the joint distribution implied by this data generating process.
The effects of interest under the full intervention $do(\mbt=\vts)$ are the average and~\emph{\condeff}
\begin{align}
  (\text{average})\quad &\tau(\vts)=\E_{\mbz, \mbeta}f(\vts, \mbz, \mbeta)
  \quad \quad \quad 
   (\text{conditional}) \quad \phi(\vts,\vz)=\E_{\mbeta}\left[f(\vts, \vz, \mbeta)\right].
\end{align}
With observed confounders, two assumptions make causal estimation possible: \emph{ignorability} and \emph{positivity}.
Ignorability means that \emph{all} confounders $\mbz$ are observed in data.
Conditioning on all the confounders, the outcome under an intervention is distributed as if conditional on the value of the intervention:
$p(\mby = y_1 \g do(\mbt=\vts), \mbz=\vz) = p(f(\vts, \vz, \mbeta) = y_1) = p(\mby = y_1 \g \mbt=\vts, \mbz=\vz) $.
This allows the expression of \avgeff\ as an expectation over the \emph{observed} outcomes
$\tau(\vts) = \E_{\mbz, \mbeta}[f(\vts, \mbz, \mbeta)] = \E_{\mbz}\E[y \g \mbz, \vts].$
The conditional expectation only exists for all $\vts$ if $p(y \g \mbz, \mbt=\vts) = \nicefrac{p(y, \mbz, \mbt=\vts)}{p(\mbz) p(\mbt =\vts \g \mbz)}$ exists.
\emph{Positivity} guarantees this existence 
\begin{align}
(\text{positivity})\quad \forall \vts \in \supp(\mbt) \quad p(\mbz=\vz)>0\implies  p(\mbt =\vts \g \mbz = \vz) > 0.
\label{eq:positivity}
\end{align}

\section{\Acrlong{efc}}\label{sec:est-without-positivity}
In \tci, causal estimation relied on knowing the confounders.
In this section, we consider settings where confounders are known via a function of the pre-outcome variables $h(\mbt)=\mbz$. 
We call this setting~\glsreset{efc}\emph{\gls{efc}}.
An example of this is~\gls{gwas}, where \glspl{snp} (the pre-outcome variables) are used to estimate the confounding population structure through methods like \gls{pca}~\citep{yu2006unified}. 
Assuming the confounders are a function of the pre-outcome variables violates positivity in general.
Positivity is violated in this setting because
\[\forall t_1, t_2 \in \supp(\mbt) \,\,\, s.t. \,\,\, h(t_2) \not = h(t_1)\,\, \implies p(\mbz = h(t_2) \g \mbt=t_1) = 0 \not = p(\mbz = h(t_2)) > 0\]
In words, two different confounder values cannot occur for the same $t$.
A positivity violation precludes nonparametric effect estimation of the full intervention $do(\mbt=\vts)$.

\paragraph{Positivity and Regression Identifiability}
Positivity can be viewed as providing identifiability. 
To see this, let the confounder be $\mbz=h(\mbt)$ and the outcome be
$y(\mbt, \mbz, \mbeta) = \mbz + h(\mbt).$
Now consider regressing $\mbz$ and $\mbt$ onto $y$.
Then, functions $y = \alpha \mbz + \beta h(\mbt)$ indexed by $\alpha,\beta$, such that $\alpha + \beta=2$, are consistent with the observed data.
Thus, there exist infinitely many solutions to the conditional expectation of $y$ on $(\mbt,\mbz)$, meaning that the regression is not identifiable.
Assuming positivity necessitates sufficient randomness to identify the regression and thus the causal effect.
A violation of positivity means that nonparametric estimation of causal effects needs further assumptions.

\subsection{Setup for \gls{efc}}
In \gls{efc}, the confounder is provided as a non-bijective function $h$ of the
pre-outcome variables $\mbt$.
To reflect this property, we use $h(\mbt)$ to denote the confounder.
As an illustrative example, let $\cG$ be the Gamma distribution and consider $\mbz\in\{-1,1\}, p(\mbz=1) = 0.5$ is the confounder and the intervention of interest is $\mbt = \mbz*\cG(1,\exp(\mbz))$.
Note $\sign(\mbt) = \mbz$ meaning that $h(\mbt) = \sign(\mbt)$ is the confounder.
\Cref{fig:clarify-ht} shows causal graphs connecting our \gls{efc} notation to that in \tci.
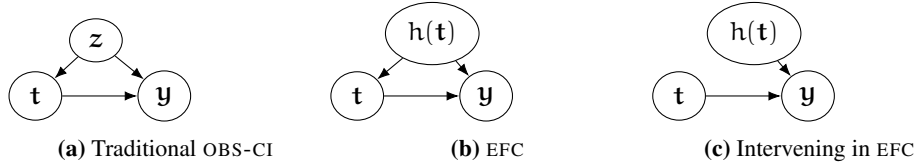
\begin{figure}[t]  
\centering
  \begin{subfigure}[b]{0.3\textwidth}
    \begin{tikzpicture}
    \node[state] (2) {$\mbt$};
    \node[state] (3) [right =of 2] {$\mby$};
    \node[state] (1) [above right =of 2,xshift=-0.7cm,yshift=-0.7cm] {$\mbz$};

    \path (1) edge node[above] {} (2);
    \path (2) edge node[above] {} (3);
    \path (1) edge node[above] {} (3);
\end{tikzpicture}
    \caption{\Tci} \label{fig:M1}  
  \end{subfigure}
\begin{subfigure}[b]{0.3\textwidth}
    \begin{tikzpicture}
    \node[state] (2) {$\mbt$};
    \node[state] (3) [right =of 2] {$\mby$};
    \node[state] (1) [above right =of 2,xshift=-0.7cm,yshift=-0.7cm] {$h(\mbt)$};

    \path (1) edge node[above] {} (2);
    \path (2) edge node[above] {} (3);
    \path (1) edge node[above] {} (3);
\end{tikzpicture}
\caption{\gls{efc}} \label{fig:M2}  
\end{subfigure}
\begin{subfigure}[b]{0.3\textwidth}
    \begin{tikzpicture}
    \node[state] (2) {$\mbt$};
    \node[state] (3) [right =of 2] {$\mby$};
    \node[state] (1) [above right =of 2,xshift=-0.7cm,yshift=-0.7cm] {$h(\mbt)$};

    \path (2) edge node[above] {} (3);
    \path (1) edge node[above] {} (3);
\end{tikzpicture}
\caption{Intervening in \gls{efc}} \label{fig:M2}  
\end{subfigure}
\caption{Causal Graphs for \Tci\ vs. \gls{efc}.}
\label{fig:clarify-ht}
\vspace{-0.4cm}
\end{figure} 
With noise $\mbeta \sim p(\mbeta) (\mbeta\indep \mbt)$, our causal model samples, in order, the confounder "part" of pre-outcome variables $h(\mbt)$, the pre-outcome variables $\mbt$, and the outcome $\mby$ via the \emph{outcome function} $f$
\footnote{We also assume no interference~\cite{hernan2020causal} (also called Stable Unit Treatment Value Assumption~\cite{rubin1980randomization}) which means that an individual's outcome does not depend on others' treatment.
In \gls{efc}, when $\mbt$ and $\mbeta$ are sampled IID there is no interference.
To see this, note $\forall i,j\,\, (\mbt_i, \mbeta_i)\indep (\mbt_j, \mbeta_j)  \implies (\mby_i ,\mbt_i) \indep (\mby_j, \mbt_j) \implies \mby_i \indep \mbt_j$.}:
\[h(\mbt) \sim p(h(\mbt)) \quad \mbt~\sim p(\mbt \g h(\mbt)) \quad \mby = f(\mbt, h(\mbt), \mbeta)\]
Similar to \tci, for an intervention $\vts$ the average effect, $\tau(\cdot)$, and the conditional effect, $\phi(\cdot, \cdot)$ at $h(\vts_2)$, respectively, are defined as:
\begin{align}\label{eq:full-effs}
  \quad \tau(\vts)=\mathop{\E}_{h(\mbt), \mbeta}[f(\vts, h(\mbt),\mbeta)], 
   \quad \quad \quad \phi(\vts,h(\vts_2))=\mathop{\E}_{\mbeta}[f(\vts, h(\vts_2),\mbeta)].
\end{align}
As the pre-outcome variables determine the confounder, positivity is violated. 
Further, the \emph{outcome function} $f(\mbt, h(\mbt), \mbeta)$ could recover the exact value of $h(\mbt)$ from $\mbt$ instead of its second argument.
Thus, two different outcome functions could lead to the same observational data distribution, posing a fundamental obstacle to causal effect estimation.
This is the central challenge in \gls{efc}.

\subsection{Causal Questions With Functional Positivity}\label{sec:qs-without-pos}
Without positivity, we can only estimate the effects of certain functions of $\mbt$.
We call such interventions, on some function $g(\mbt)$, \emph{\funcints}.
The implied causal model for the outcome for functional intervention
value $g(\vts)$ and confounder value $h(\vts_2)$
is first $\mbt \sim p(\mbt \g g(\mbt) = g(\vts), h(\mbt) = h(\vts_2))$ and then $\mby = f(\mbt, h(\vts_2), \mbeta)$
\footnote{Intervening on $g(\mbt)$ can be interpreted as making a \emph{soft intervention}~\citep{eberhardt2007interventions,correa2020calculus} of $\mbt$ to $p(\mbt\g \mbz, g(\mbt)=g(\vt))$.
}.
Then, the \textit{functional} average effect is
\begin{align*}
  (\text{average})\quad &\tau(g(\vts))={\E}_{h(\mbt), \mbeta}\E_{\mbt \g g(\mbt) = g(\vts), h(\mbt)}[f(\mbt, h(\mbt), \mbeta)].
\end{align*}
An example of a functional intervention is intervening on the cumulative dosage of a drug.
In contrast, traditional interventions would set each individual dose given at different points in time.

\paragraph{\gls{func-pos} and Functional Effect Estimation}
For the causal model above to be well-defined for all \funcints~$g(\vts)$, the conditional $p(\mbt \g g(\mbt) = g(\vts), h(\mbt) = h(\vts_2))$ must exist.
To guarantee this existence, we define \emph{\glsreset{func-pos}\gls{func-pos}} for any $g(\vts)$
\begin{align}\label{eq:functional-positivity}
(\text{\gls{func-pos}})\quad p(h(\mbt) = h(\vts_2))>0 \Longrightarrow  p(g(\mbt) = g(\vts) \g h(\mbt) = h(\vts_2)) > 0 .
\end{align}
\Gls{func-pos} says that the function of the pre-outcome variables that is being intervened on needs to have sufficient randomness when the function of the pre-outcome variables that defines the confounders is fixed.
Further, under \gls{func-pos}, effect estimation for \funcints\ is reduced to \tci\ on data $p(\mby, g(\mbt), h(\mbt))$. 
With positivity and ignorability satisfied, traditional causal estimators such as propensity scores \citep{rosenbaum1983central}, matching \citep{ratkovic2014balancing}, regression~\citep{hill2010bayesianci}, and doubly robust methods \citep{robins2000robust} can be used to estimate the causal effect.
Focusing on regression, let $f_\theta$ be a flexible function, then
$\min_{\theta} \E_{y, \mbt} [(\mby - f_\theta(h(\mbt), g(\mbt)))^2]$
would estimate the conditional expectation of interest : $\E[\mby\g h(\mbt), g(\vts)]$.
With $\theta$, the effect of $g(\vts)$ can be estimated by averaging the estimate of the conditional expectation over the marginal distribution $p(h(\mbt))$:
\begin{align}
  \tau(g(\vts)) = \E_{\mbt}[f_\theta(h(\mbt), g(\vts))].
  \label{eq:causal-estimate-positivity}
\end{align}

\section{Identification of effects of the full intervention}\label{sec:full-est}
When positivity is violated, causal effects cannot be estimated as conditional expectations over the observed data in general.
We give a functional condition, called \glsreset{c-red}\gls{c-red}, that allows us to estimate the effect of the full intervention $do(\mbt=\vts)$, even when positivity is violated.
Specifically, \gls{c-red} allows us to construct a \textit{surrogate intervention} $\vtt(\vts, h(\vts_2))$ whose conditional effect at $h(\vtt)$ matches the conditional effect of interest, $\phi(\vts, h(\vts_2))$.
Let $\vt$ be a fixed value of the full intervention, then \gls{c-red} is
\begin{assumption*}\label{assn:c-redundancy}
  Recall the outcome $y = f(\vt, h(\vt), \eta)$. With $\nabla_\vt$ as gradient w.r.t.\ to argument $\vt$:
  \[\forall \vt, h(\vt_2), \eta, \quad  \nabla_{\vt} f (\vt, h(\vt_2), \eta)^T\nabla_{\vt} h(\vt) = 0.\]
\end{assumption*}
In words, \gls{c-red} is the condition that the outcome function $f$ uses the value of the confounder from its second argument instead of computing $h(\mbt)$ from the first argument\footnote{If $f$ transforms its first argument $\vt$ into $h(\vt)$ as one amongst many different computations, the chain rule implies $\nabla_\vt f(\vt, h(\vts_2))^\top \nabla_\vt h(\vt)$ has a term $\|\nabla_\vt h(\vt)\|^2 $ which is non-zero in general.}.
To compute the conditional effect $\phi(\vts, h(\vts_2))$, we develop \glsreset{ours}\gls{ours}.
\gls{ours}'s key step is to construct a surrogate intervention $\vtt(\vts,h(\vts_2))$ such that
  \begin{align*}
    \begin{split}
      \phi(\vts, h(\vts_2)) &=  \phi({\vtt(\vts, h(\vts_2))}, h(\vts_2)), \quad \quad 
        h(\vts_2) =   h(\vtt(\vts, h(\vts_2))) .
    \end{split}
  \end{align*}

\begin{wrapfigure}[9]{r}{0.32\textwidth}
   \vspace{-25pt}
   \centering
   \includegraphics[width=0.3\columnwidth]{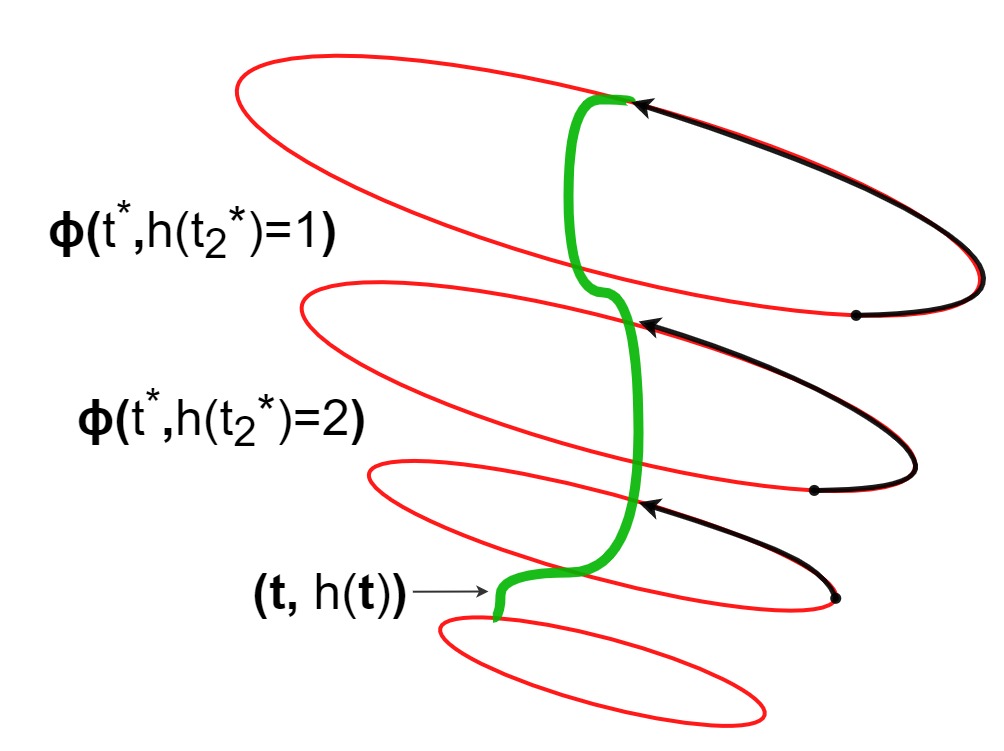}
   \caption{\gls{ours}'s traversal.}
   \label{fig:level-sets}
\end{wrapfigure}
By definition, a surrogate intervention lives in the conditional effect level-set: $\{\vt: \phi(\vt, h(\vts_2))= \phi(\vts, h(\vts_2))\}$.
So \gls{ours} searches this level-set for $\vtt(\vts, h(\vts_2))$.
See~\cref{fig:level-sets} which plots the conditional effect level-sets with the value of $h(\mbt)$ fixed (\textcolor{red}{red}) in $(\supp(\mbt), \supp(h(\mbt)))$-space.
\textcolor{green}{Green} corresponds to the observed data, $\supp(\mbt, h(\mbt))$.
\gls{ours} finds $\vtt(\vts, h(\vts_2))$ by traversing the level-sets (\textbf{black}) to account for the confounder part mismatch $h(\vts)\not=h(\vts_2)$.
\gls{c-red} ensures \gls{ours} can traverse these level-sets as it implies $\nabla_{\vt}\phi(\vt, h(\vt_2))\nabla_{\vt}h(\vt) = 0$ under the regularity conditions in~\cref{thm:thm-c-red}.
Thus, under \gls{c-red}, surrogate interventions can be constructed by solving a gradient flow equation which guarantees identification as follows:
\newcommand{\lemmaone}{

  Assume \gls{c-red} holds. Assuming the following:
  \begin{enumerate}[leftmargin=0.5cm]
    \item Let $\vtt(\vts, h(\vts_2))$ be the limiting solution to the gradient flow equation $\fullder{\vt(s)}{s} = -\nabla_{\vt}(h(\vt(s)) - h(\vts_2))^2$, initialized at $\vt(0)=\vts$; i.e. $\vtt(\vts, h(\vts_2)) = \lim_{s\rightarrow \infty } \vt(s)$.\\
    Further, let $h(\vtt(\vts, h(\vts_2))) = h(\vts_2)$
    and $\vtt(\vts, h(\vts_2))\in \supp(\mbt)$.
    \item $f(\vt, h(\vt), \eta)$ and $h(\vt)$ as functions of $\vt, h(\vt)$ are continuous and differentiable and the derivatives exist for all $\vt,\eta$.
    Let $\nabla_{\vt} f(\vt, h(\vt), \eta)$ exist and be bounded and integrable w.r.t.\ the probability measure corresponding to $p(\mbeta)$, for all values of $\vt$ and $h(\vt)$.
  \end{enumerate}
  Then the conditional effect (and therefore the average effect) is identified:
\begin{align}\label{eq:effect-estimates}
  \begin{split}
  \phi(\vts,h(\vts_2)) = \phi\left(\vtt(\vts, h(\vts_2)), h(\vtt(\vts, h(\vts_2)))\right)= \E\left[\mby\g \mbt = {\vtt(\vts, h(\vts_2))}\right]
  \end{split}
  \end{align} 
  }
\begin{theorem}\label{thm:thm-c-red}
  \lemmaone{}
\end{theorem}
In words, the key idea is that starting at $\vt(0) = \vts$ and following $\nabla_\vt h(\vt)$ means $\vt(s)$ always lies in the level-set $\{\vt : \phi(\vt, h(\vts_2)) = \phi(\vts, h(\vts_2))\}$.
See~\cref{appsec:thm-c-red} for the proof.
While \gls{c-red} is stated in terms of the gradient of the outcome function, it suffices for~\cref{thm:thm-c-red} to assume a weaker condition about the gradient of the conditional effect: $\nabla_\vt \E_{\mbeta} f(\vt, \vt_2, \mbeta)^\top \nabla_{\vt} h(\vt) = 0$.

\paragraph{Surrogate Positivity}
In~\cref{thm:thm-c-red}, we assumed that the surrogate $\vtt(\vts, h(\vts_2)) \in \supp(\mbt)$.
This condition, which we call surrogate positivity (analogous to positivity), states that for any intervention and confounder, surrogate interventions that are limiting solutions to the gradient flow equation have nonzero density conditional on the confounder value.
Formally, for any intervention $\mbt =\vts$
\begin{align}\label{eq:surr-pos}
  p(h(\mbt) = h(\vts_2)) > 0 \implies p(\mbt = \vtt(\vts, h(\vts_2)) \g h(\mbt) = h(\vts_2)) > 0,
\end{align}
and $\vtt(\vts, h(\vts_2))$ satisfies assumption 1 in~\cref{thm:thm-c-red}.
Surrogate positivity  along with \gls{c-red}, is sufficient for full effect estimation under \gls{efc}.
Next, we show that the positivity assumption in traditional causal inference is a special case of surrogate positivity.

 \paragraph{\glsreset{ci}\Tci and \gls{ours}}
 
 Let the confounder and intervention of interest in \tci be $\mbz$ and $\mba$ respectively. Assume both are scalars and ignorability and positivity hold.
 This setup can be embedded in \gls{efc} by defining the vector of pre-outcome variables as:
 $ 
   \mbt = [\mba; \mbz].
 $
 In this setting, \gls{c-red} and surrogate positivity(\cref{eq:surr-pos}) hold by default.
 Let the outcome be $\mby  = f(\mbt, h(\mbt)) = f(\mba, \mbz)$, where $f$ only depends on the first element of $\mbt$, i.e.\ $\mba$\footnote{We ignore noise in the outcome for ease of exposition.}.
 Let $e_1 = [1, 0]$ and $e_2 = [0, 1]$.
 In \tci as \gls{efc}, $\nabla_{\vt} f(\vt, h(\vts_2)) \propto e_1$ and $\nabla_{\vt} h(\vt) \propto e_2$ meaning that $\nabla_{\vt} f(\vt, h(\vts_2)) ^\top \nabla_{\vt} h(\vt) = 0$.
 Thus, \gls{c-red} holds by default.
 Moreover, under positivity of $\mba$ w.r.t.\ $\mbz$, we also have surrogate positivity for \tci as an \gls{efc} problem.
In this setting, \gls{ours} computes $\vtt=[a^*, h(\vts_2)]$ by following $-\nabla_{\vt} h(\vt) = [0, -1]$, which only changes the value of $h(\vt_2)$, not the value of $a$.
 Thus, $\vts$ and $\vtt(\vts, h(\vts_2))$ will have the same first element and $\vtt$'s second element will be $h(\vts_2)$. As $\mba$ has positivity w.r.t.\ $\mbz$, we have $p(\mba = a^* , \mbz=h(\vts_2))>0$ which means $\vtt\in\supp(\mbt)$. The estimated conditional effect is $\E[\mby \g \mbt = \vtt(\vts, h(\vts_2))] = f([a^*, z^*], h(\vts_2)) = \E[\mby \g \mba = a^*, \mbz=h(\vts_2))]$, which matches the estimate in \tci.

 \paragraph{Implementation of \gls{ours}}
 \Gls{ours} first estimates the conditional expectation $\E[\mby\g \mbt]$; this can be done with model-based or nonparametric estimators.
 This is achieved by regressing $\mby$ on $\mbt$, $\hat{f} = \argmin_{u\in \mathcal{F}} \E_{\mby, \mbt \sim D} (\mby - u(\mbt))^2 $, with empirical distribution $D$.
 The surrogate intervention $\vtt(\vts, h(\vts_2))$ is computed using Euler integration to solve the gradient flow equation.
 Euler integration in this setting is equivalent to gradient descent with a fixed step size.
 Other, more efficient schemes like Runge–Kutta numerical integration methods~\cite{ascher1998computer} could also be used.
 The conditional effect estimate is $\hat{f}(\vtt(\vts, h(\vts_2)))$. See~\cref{alg:full-eff-estimation} for a description.

\subsection{Estimation error of \gls{ours} in practice}\label{sec:estimation-error}
To compute the surrogate intervention $\vtt$, \gls{ours} uses the gradients of $h(\cdot)$ in Euler integration.
In practice, taking Euler integration steps, instead of solving the gradient flow exactly, could result in errors.
Then $\vtt$ could lie outside the level-set of the conditional effect $\phi(\vts, h(\vts_2)) =  \E_{\mbeta}[f(\vts, h(\vts_2), \mbeta)]$.
Further, if $h(\vtt(\vts, h(\vts_2)))\not = h(\vts_2)$, \gls{ours} incurs error for conditioning on a value of the confounder that is different from $h(\vts_2)$.
The error due to $\vtt$ estimation is decoupled from the error in the estimation of $\E[\mby \g \mbt]$ which adds without further amplification.
We formalize this error:
\newcommand{\theoremmain}{
  Consider the conditional effect $ \phi(\vts, h(\vts_2))$.
  Let $\vthat(\vts, h(\vts_2))$ be the estimate of the surrogate intervention computed by \gls{ours}, computed via Euler integration of the gradient flow $\fullder{\vt(s)}{s} = - \nabla_\vt (h(\vt(s)) - h(\vts_2))^2$, initialized at $\vt(0) = \vts$.
  Assume the true surrogate $\vtt(\vts, h(\vts_2))$ exists and is the limiting solution to the gradient flow equation.
  \begin{enumerate}
    \item Let the finite sample estimator of $\E[\mby\g \mbt=\vt]$ be  $\hat{f}(\vt)$. 
    Let  
      the error for all $\vt$ be bounded, $|\hat{f}(\vt) - \E[\mby\g \mbt=\vt] |\leq c(N)$, where $N$ is the sample size and $ \lim_{ N\rightarrow \infty } c(N) = 0 $.
    \item Assume $K$ Euler integrator steps were taken to find the surrogate estimate $\vthat(\vts, h(\vts_2))$, each of size $\ell$.
        Let the maximum confounder mismatch be $\max_{i\leq K}(h(\vt_i) - h(\vts_2))^2 = M$.
    \item Let $L_{z,\vt}$ be the Lipschitz-constant of $\phi(\vt, h(\vt_2))$ as a function of $h(\vt_2)$, for fixed $\vt$.\\
     Let $L_e$ be the Lipschitz-constant of $\E[\mby \g \mbt=\vt] = \phi(\vt, h(\vt))$ as a function of $\vt$.
     \\
      Assume $h$ has a gradient with bounded norm, $\|\nabla h(\vt)\|_2 < L_h$.
    \\
    Assume $f$'s Hessian has bounded eigenvalues: $\forall \vt, \vt_2,\,\, \| \nabla^2_\vt \phi(\vt, h(\vt_2))\|_2 \leq \sigma_{\mathtt{H} \phi} $.
  \end{enumerate}
  The conditional effect estimate error, $\xi(\vts,h(\vts_2)) = |\hat{f}(\vthat) - \phi(\vts, h(\vts_2))|$, is upper bounded by:
  \begin{align}\label{eq:est-err-bound}
    c(N) +   \min\left(L_e\|\vtt - \vthat\|_2, 
                          \,\, 2K \ell^2 \left(\mathcal{O}(\ell) + M \sigma_{\mathtt{H} \phi } L^2_{h}  \right) +  L_{z,\vthat} \|h(\vthat) - h(\vts_2)\|_2\right)
  \end{align}
}
\begin{theorem}\label{thm:main-effect-bound}
\theoremmain
\vspace{2pt}
\end{theorem}
See \cref{appsec:lode-error} for the proof.
\Cref{thm:main-effect-bound} captures the trade-off between biases due to conditioning on the wrong confounder value and due to the accumulated error in solving the gradient flow equation.
This accumulated error analysis may be loose in settings where the sum of many gradient steps lead to $\vthat \approx \vtt$, even if each step individually induces large error.
In such settings, the term that depends on $\|\vthat - \vtt\|_2$ is a better measure of error.
The maximum-mismatch $M$ appears because Euler integrator takes steps that depend on the magnitude of the gradient which depends on the mismatch value $(h(\vt_i) - h(\vts_2))$.
If mismatch is large for some $i$, the Euler step could lead to a large error for a fixed step size $\ell$.
We discuss the assumptions in~\cref{thm:thm-c-red,thm:main-effect-bound} in~\cref{appsec:assumptions}

\subsection{Effect Connectivity and the Existence of $\vtt(\vts, h(\vts_2))$}
The key element in \Cref{thm:thm-c-red} is the surrogate intervention $\vtt$ such that its conditional effect given $h(\vtt)$, equals that of $\vts$ and $ h(\vts_2)$.
The orthogonality $\nabla_{\vt} f^\top \nabla_{\vt} h =0 $, is a functional condition that does not guarantee $\vtt(\vts, h(\vts_2))$ exists in $\supp(\mbt)$; a necessity to compute $\E[\mby\g \mbt = \vtt]$ without additional parametric assumptions.
We give a general condition called \emph{Effect Connectivity} that guarantees the surrogate intervention exists.
With conditional effect $\phi(\vts, h(\vts_2))$, for any $\vts$
\begin{align}\label{eq:weird-func-pos-condition}
p(h(\mbt)=h(\vts_2))>0 \implies p(\phi(\mbt, h(\mbt)) = \phi(\vts,h(\vts_2))\g h(\mbt) = h(\vts_2)) > 0.
\end{align}
In words, $\mbt$ has a chance of setting the conditional effect to any possible value $\supp( \phi(\mbt, h(\mbt_2)))$ given any confounder value $h(\vts_2) \in \supp(h(\mbt))$.
An equivalent statement is that every level set of the conditional effect $\phi(\vts, h(\vts_2))$, with $h(\vts_2)$ fixed, contains an intervention for each confounder value.
That is, for some $h(\vts_2)$ define the level set $A_c = \{\vts; f(\vts, h(\vts_2))=c\}$, then $\forall h(\vts_2)\in \supp(h(\mbt)),\, p(\mbt \in A_c \g h(\mbt) =h(\vts_2) ) > 0$.
\newcommand{\sufficiencylemma}{
 Under Effect Connectivity, \cref{eq:weird-func-pos-condition}, any surrogate intervention $\vtt(\vts,h(\vts_2)) \in \supp(\mbt)$.
  }
\begin{theorem}\label{thm:thm-sufficiency}
  \sufficiencylemma{}
\end{theorem}
  We give the proof in~\cref{appsubsec:sufficiency}.
  Whether the intervention $\vtt(\vts, h(\vts_2))$ can be found via tractable search is problem-specific.
  If the surrogate $\vtt(\vts, h(\vts_2))$ exists $\forall \vts,h(\vts_2)$, then \cref{eq:weird-func-pos-condition} holds by definition of the surrogate.
  Effect Connectivity allows us to reason about values of $f$ anywhere in $\supp(\mbt)\times\supp(h(\mbt))$ using only samples from $p(\mby, \mbt)$.
  Further, it is necessary in \gls{efc}:
  \newcommand{\necessitylemma}{
  Effect Connectivity is necessary for nonparametric effect estimation in \gls{efc}.
   }
 \begin{theorem}\label{thm:thm-necessity}
   \necessitylemma{}
 \end{theorem}
  We prove this in~\cref{appsubsec:necessity-proof}.
  Effect Connectivity ensures that causal models with different causal effects have different observational distributions.
  Then, parametric assumptions on the causal model are not necessary to estimate effects.

\section{Experiments}\label{sec:exps}
We evaluate \gls{ours} on simulated data first and show that \gls{ours} can correct for confounding.
We also investigate the error induced by imperfect estimation of the surrogate intervention in \gls{ours}.
Further, we run \gls{ours} on a \gls{gwas} dataset~\citep{wellcome2007genome} and demonstrate that \gls{ours} is able to correct for confounding and recovers genetic variations that have been reported relevant to Celiac disease~\citep{dubois2010multiple,sollid2002coeliac,hunt2008novel,adamovic2008association}.
\subsection{Simulated experiments}
We investigate different properties of \gls{ours} on simulated data where ground truth is available.
Let the dimension of $\mbt$ (pre-outcome variables) be $T=20$ and outcome noise be $\mbeta \sim \cN(0,0.1)$.
We consider two \gls{efc} causal models, denoted by $A$ and $B$ with different $h(\mbt)$ and $f(\mbt, h(\mbt), \mbeta)$:
\begin{align*}
  (A) \quad &  h(\mbt) = \gamma \frac{\sum_i \mbt_i}{\sqrt{T}},\quad \quad \mbt\sim \cN(0, \sigma^2\mathbb{I}^{T\times T}),
    \quad y = \frac{\sum_i(-1)^i \mbt_i}{\sqrt{T}} +  \alpha h(\mbt)^2 + (1+\alpha)h(\mbt) +  \mbeta\\
  (B) \quad  &  h(\mbt) = \mathop{\textstyle \sum}_{ i:i \in 2\mathbb{Z}} \gamma \mbt_i\mbt_{i+1},\quad \mbt\sim \cN(0, \sigma^2\mathbb{I}^{T\times T}),
      \quad y = \frac{\sum_i (-1)^i \mbt^2_i}{\sqrt{T}} +  \alpha h(\mbt)+ \mbeta
\end{align*}
In both causal models, \gls{c-red} is satisfied.
The constant $\gamma$ controls the strength of the confounder and the constant $\alpha$ controls the Lipschitz constant of the outcome as a function of the confounder.
We let the variance $\sigma^2 = 1$, unless specified otherwise.
In the following, we train on $1000$ samples and report conditional effect \gls{rmse}, computed with another $1000$ samples.
We used a degree-2 kernel ridge regression to fit the outcome model as a function of $\mbt$.
This model is correctly specified, and so the conditional $\E[\mby\g \mbt = \vt]$ can be estimated well.
We compare against a baseline estimate of conditional effect that is the same outcome model's estimate of $\E[\mby \g \mbt=\vts]$.
This baseline fails to account for confounding and produces a biased estimate of the conditional effect of $do(\mbt=\vts)$, conditional on any $h(\vts_2)\not=h(\vts)$.
\begin{figure*}[t!]
    \centering
    \begin{subfigure}[b]{0.45\textwidth}
        \centering
        \includegraphics[width=\textwidth]{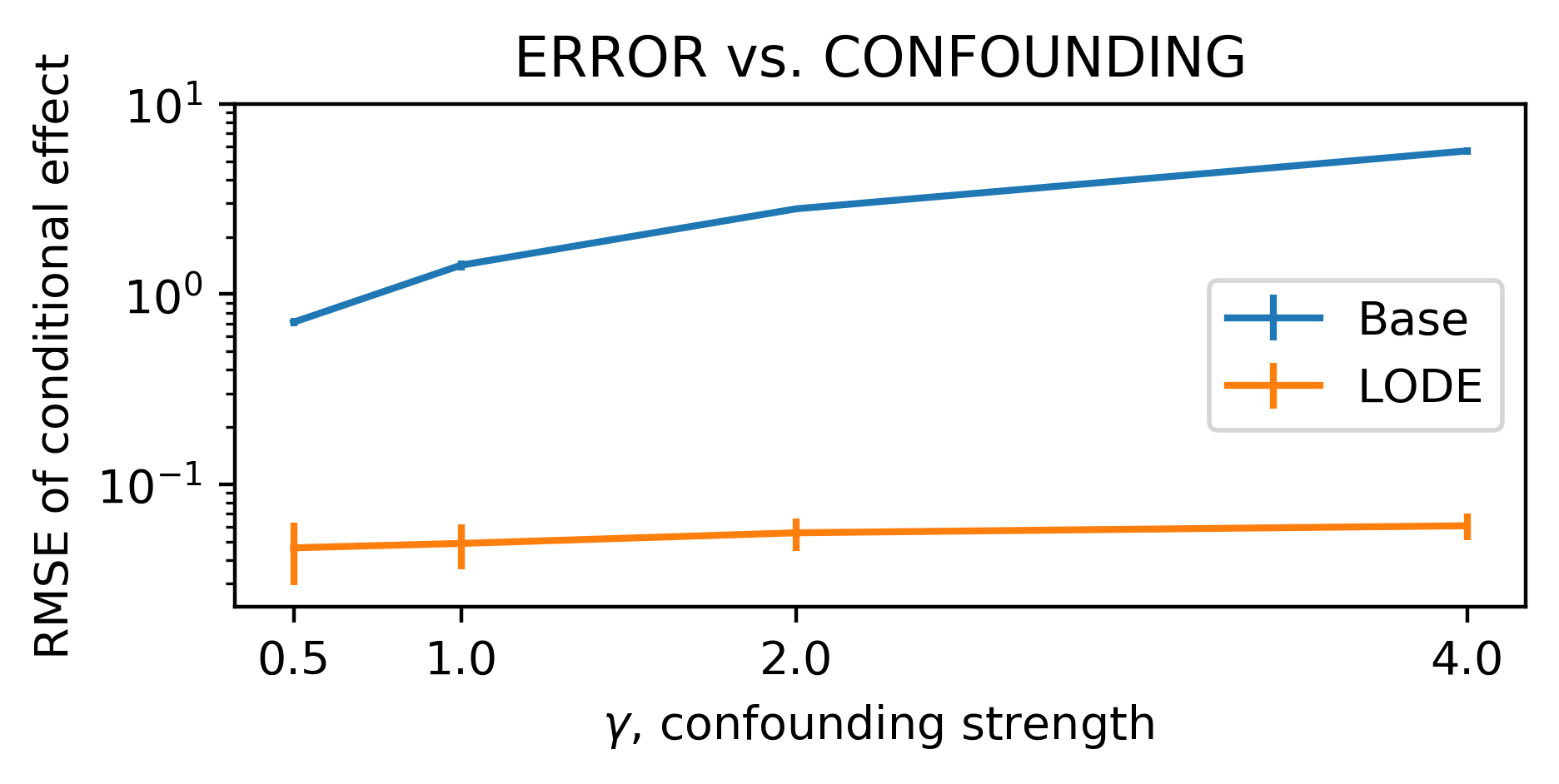}
        \caption{Causal Model $A$}
    \end{subfigure}
    ~ 
    \begin{subfigure}[b]{0.45\textwidth}
        \centering
        \includegraphics[width=\textwidth]{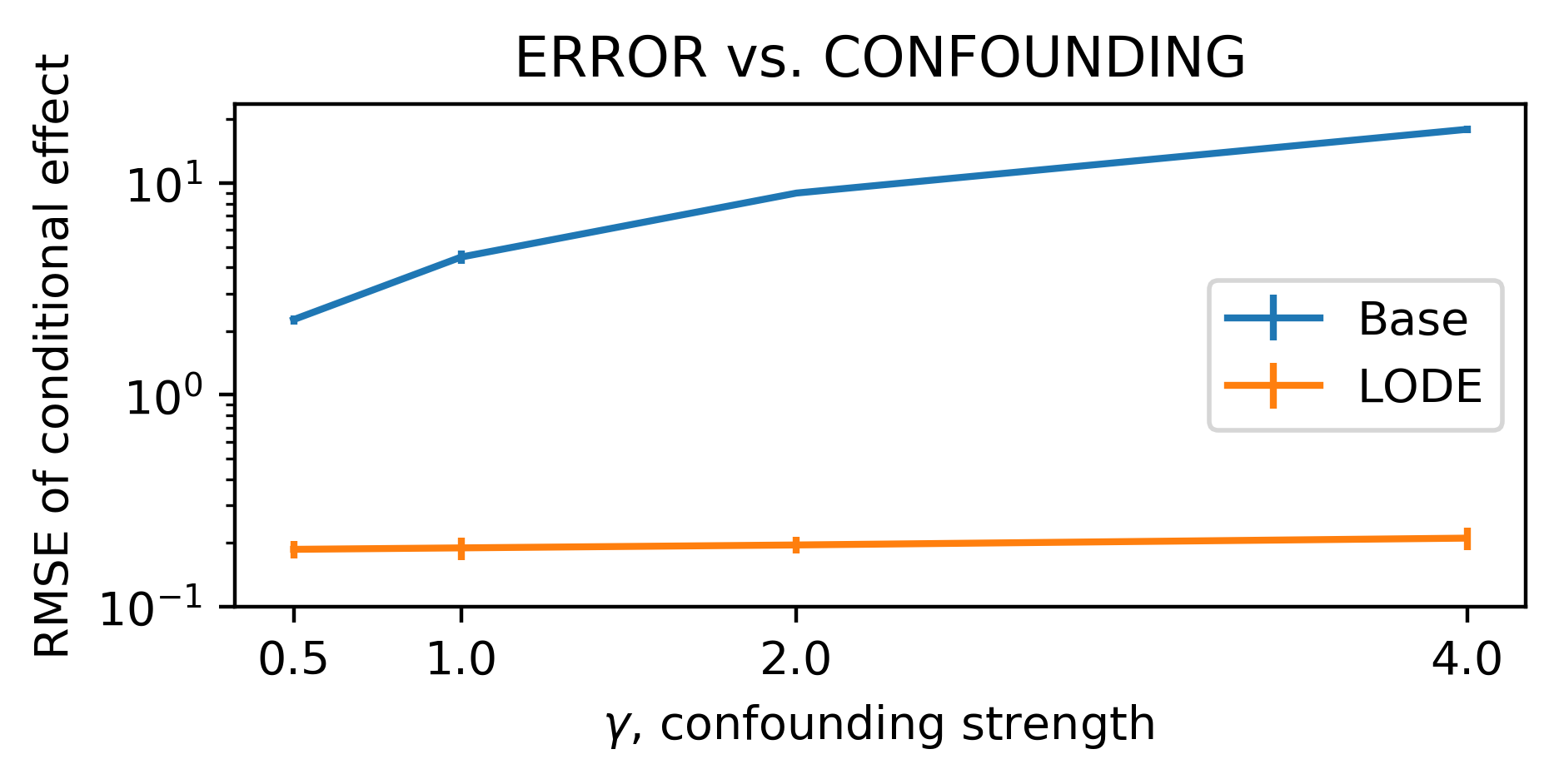}
        \caption{Causal Model $B$}
    \end{subfigure}
    \caption{\Gls{rmse} of estimated \condeff{} vs. strength of confounding $\gamma$.
    \gls{ours} corrects for confounding and produces good effect estimates across different values of $\gamma$.}
    \label{fig:lode-est}
\end{figure*}

First, we investigate how well \gls{ours} can correct for confounding for both causal models.
We let $\alpha=1$ and obtain surrogate estimates by Euler integrating until the quantity $\E_{\vts, h(\vts_2)}(h(\vt(s))- h(\vts_2))^2$ is smaller than $10^{-4}$ times value at initialization, where $\E_{\vts, h(\vts_2)}$ is expectation over the evaluation set.
In~\cref{fig:lode-est}, we plot the mean and standard deviation of conditional effect \gls{rmse} averaged over $10$ seeds, for different strengths of confounding.
We see that \gls{ours} is able to estimate effects well across multiple strengths of confounding while the baseline suffers.
\begin{figure}[t]
  \centering
  \begin{subfigure}[b]{0.45\textwidth}
      \centering
      \includegraphics[width=\textwidth]{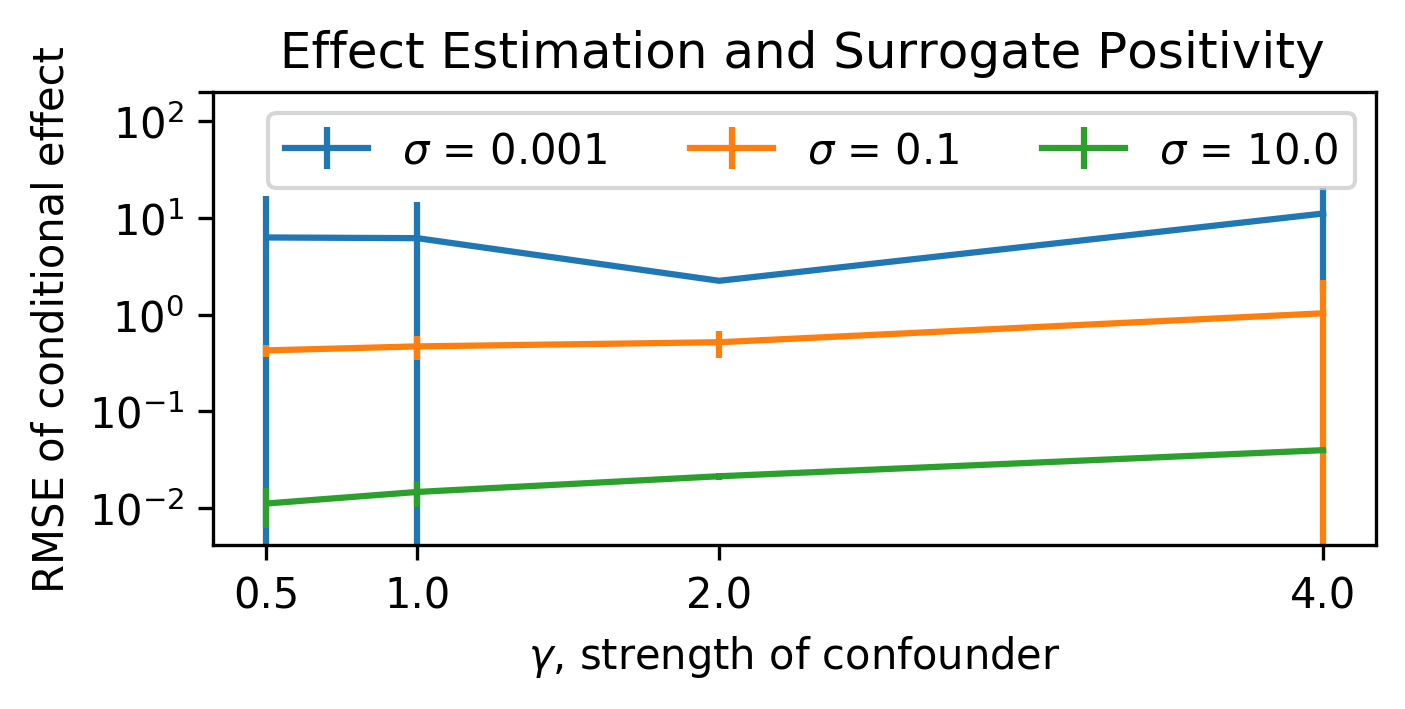}
      \caption{Causal Model $A$}
  \end{subfigure}
  ~ 
  \begin{subfigure}[b]{0.45\textwidth}
      \centering
      \includegraphics[width=\textwidth]{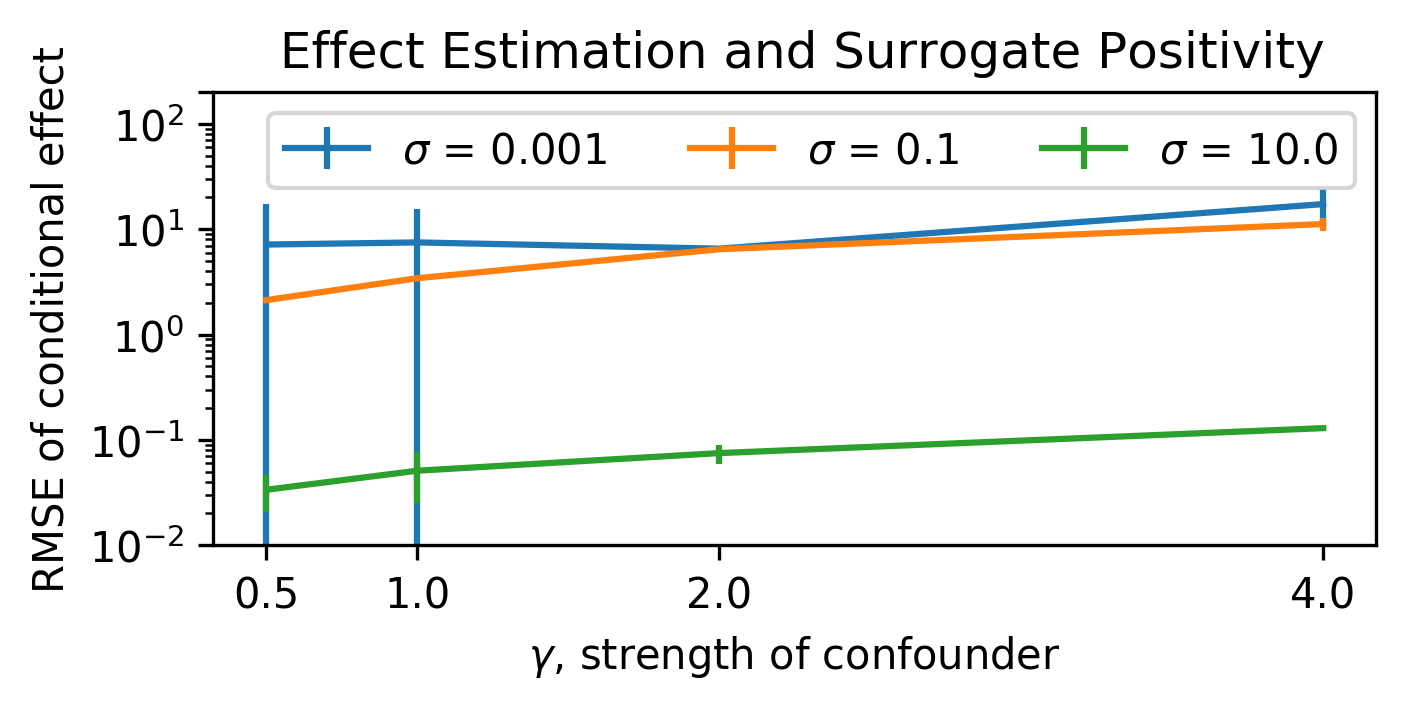}
      \caption{Causal Model $B$}
  \end{subfigure}
  \caption{\gls{rmse} of estimated \condeff{} estimate vs. the strength of confounding $\gamma$, for different levels of variance of $\mbt$, $\sigma^2$.
  Small $\sigma$ leads to large conditional estimation error.}
  \label{fig:eff-pos-experiment}
\end{figure}
\begin{wrapfigure}[14]{r}{0.42\textwidth}
  \centering
  \includegraphics[width=0.42\textwidth]{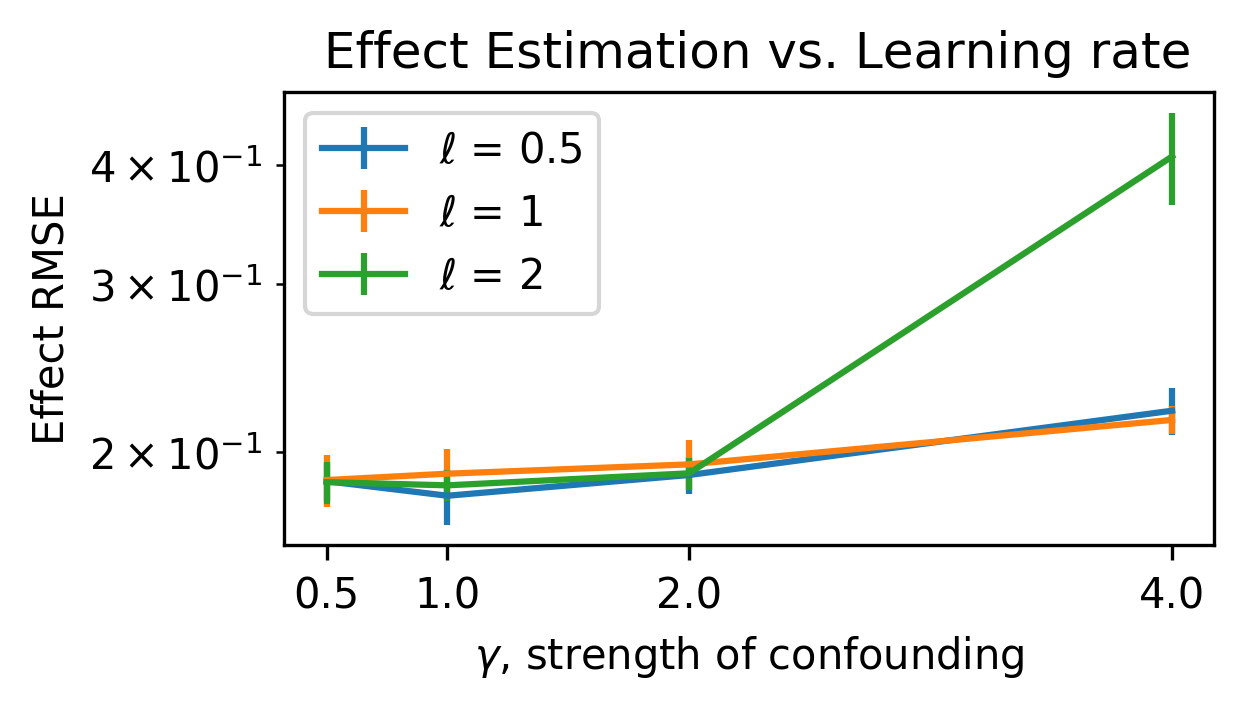}
  \caption{
    \Gls{rmse} of estimated \condeff{} vs. step size in Euler Integrator in causal model $B$.
    Accumulating error due to large step size in Euler integrator increases with strength of confounding.
  }
  \label{fig:lr-error}
  \end{wrapfigure}

Second, we investigate \gls{ours}'s estimation when surrogate positivity holds but the probability $p(\mbt \approx \vtt(\vts, h(\vts_2)))$ is very small.
This results in estimation error due to poor fitting of the outcome model in low density regions of $\supp(\mbt)$.
We run \gls{ours} on simulated data where $\mbt$ is generated with different variances ($\sigma^2$).
For small $\sigma$, the outcome model error is large when using surrogate interventions $\vtt(\vts,h(\vts_2))$, where either $h(\vts_2)$ or $\vts$ is large.
This leads to high variance effect estimation as we show in ~\cref{fig:eff-pos-experiment} for both causal models.
For various variances of $\mbt$, $\sigma^2$, we plot the mean and standard deviation of \gls{rmse} of estimated conditional effect over $10$ seeds, against different $\gamma$.

Third, we investigate the bias induced due to imperfect estimation of the surrogate intervention in \gls{ours} for both causal models.
We construct surrogate interventions $\vtt(\vts, h(\vts_2))$ by ensuring there is confounder-value mismatch $h(\vt)\not = h(\vts_2)$. We do this by interrupting Euler integration when the objective
    $\mathop{\E}_{\vts, h(\vts_2)}(h(\vtt(\vts, h(\vts_2))) - h(\vts_2))^2 = \delta^2 > 0,$
 where the $\E_{\vts, h(\vts_2)}$ is over our evaluation set upon which we estimate conditional effects.
For different $\alpha$, we plot in~\cref{fig:est-err-confmismatch} the mean and standard deviation of \gls{rmse} of estimated conditional effect over $10$ seeds, against different degrees of confounder mismatch, $\delta$.
The error due to confounder mismatch is mitigated by small $\alpha$, the Lipschitz-constant of the outcome as a function of $h(\mbt)$.
\begin{figure*}[t]
    \centering
    \begin{subfigure}[b]{0.45\textwidth}
        \centering
        \includegraphics[width=\textwidth]{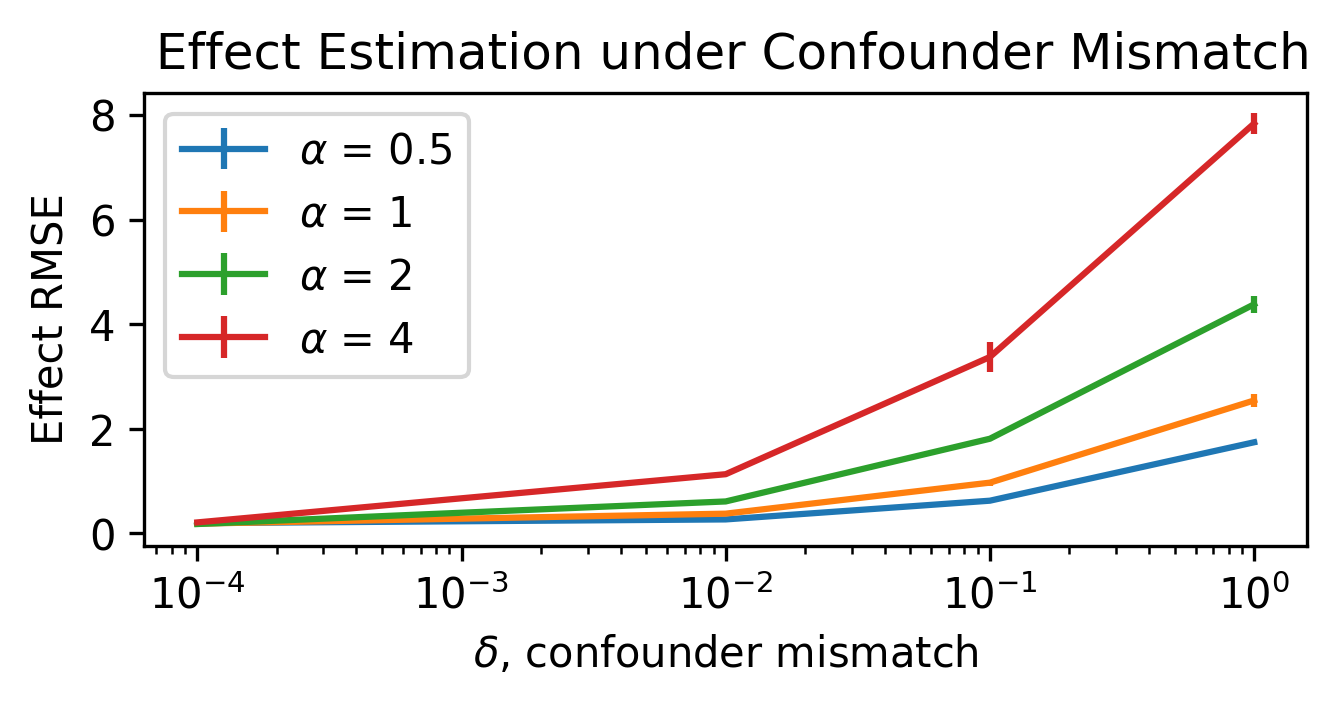}
        \caption{Causal Model $A$}
    \end{subfigure}
    ~ 
    \begin{subfigure}[b]{0.45\textwidth}
        \centering
        \includegraphics[width=\textwidth]{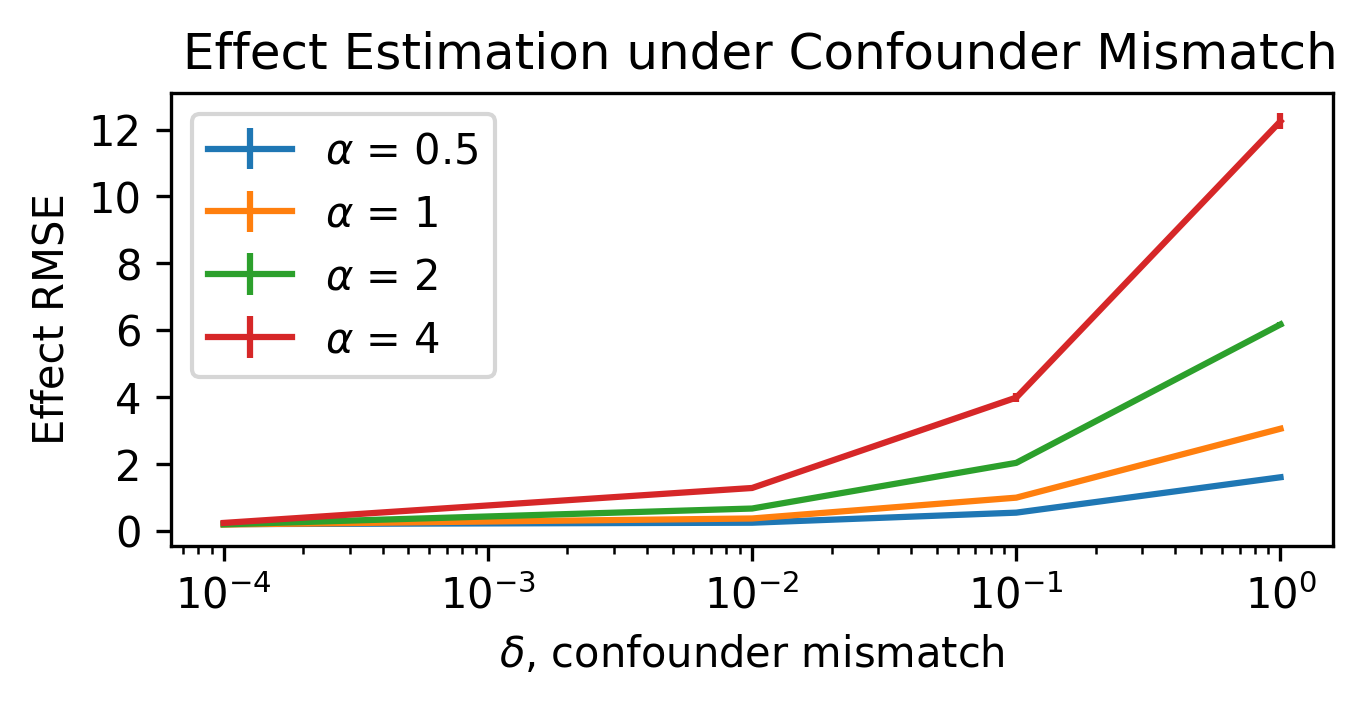}
        \caption{Causal Model $B$}
    \end{subfigure}
    \caption{\Gls{rmse} of estimated \condeff{} vs. degree of confounder mismatch $\delta$.
    Error due to conditioning on a mismatched value of the confounder increases with strength of confounding but is mitigated by smoothness of the outcome function. }
    \label{fig:est-err-confmismatch}
\end{figure*}
Finally, we consider how step size in Euler integration affects the quality of estimated effects.
Large step sizes may result in biased surrogate estimates; this bias is captured in the accumulation error in~\cref{sec:estimation-error}.
We focus on the non-linear case in causal model B where gradient errors can accumulate(see~\cref{appsec:linear-conf-func}).
We demonstrate this error in~\cref{fig:lr-error} where we plot mean and standard deviation of conditional effect \gls{rmse} against the strength of confounding, for different step sizes $\ell$.
We do not report results for larger step sizes ($\ell>2$) because Euler integration diverged for many surrogate estimates.

\subsection{Effects in Genetics (\gls{gwas})}

In this experiment, we explore the associations of genetic factors and Celiac disease.
We utilize data from the Wellcome Trust Celiac disease \gls{gwas} dataset~\citep{dubois2010multiple,wellcome2007genome} consisting of individuals with celiac disease, called cases $(n = 3796)$, and controls $(n = 8154)$.
We construct our dataset by filtering from the $\sim 550,000$ \glspl{snp}.
The only preprocessing in our experiments is linkage disequilibrium pruning of adjacent \glspl{snp} (at $0.5$ $R^2$) and PLINK~\cite{chang2015second} quality control.
After this, $337,642$ \glspl{snp} remain for $11,950$ people.
We imputed missing \glspl{snp} for each person by sampling from the marginal distribution of that \gls{snp}.
No further \gls{snp} or person was dropped due to missingness.
The objective of this experiment is to show that \gls{ours} corrects for confounding and recovers \glspl{snp} reported in the literature~\citep{dubois2010multiple,sollid2002coeliac,hunt2008novel,adamovic2008association}.
To this end, after preprocessing, we included in our data $50$ \glspl{snp} reported in ~\cite{dubois2010multiple,sollid2002coeliac,hunt2008novel,adamovic2008association} and $1000$ randomly sampled from the rest.

We use outcome models and functional confounders $h()$ traditionally employed in the \gls{gwas} literature.
We choose a linear $h(\vt) = A^\top\vt$, where $A$ is a matrix of the right singular vectors of a normalized Genotype matrix, that correspond to the top $10$ singular values~\cite{price2006principal}.
The outcome model is selected from logistic Lasso linear models with various regularization strengths, via cross validation within the training data ($60\%$ of the dataset).
We defer details about the experimental setup to~\cref{appsec:exps}.

We then use this outcome model in \gls{ours} to compute causal effects on the whole filtered dataset.
The effects are computed one \gls{snp} at a time.
  First, for each person $\vt$, create $\vt_i^1, \vt_i^0$ which correspond to the $i$th \gls{snp} set to $1$ and $0$ respectively, with all other \glspl{snp} same as $\vt$.
  Randomly sample a $h(\vts_2)$ from the marginal $p(h(\mbt))$ and, using the outcome model $P_\theta$, compute $\phi(\vt, i) = \log \nicefrac{P_\theta(y = 1\mid \vtt(\vt_i^1, h(\vts_2)))}{P_\theta(y = 1\mid \vtt(\vt_i^0, h(\vts_2)))}$.
  The average effect of \gls{snp} $i$ is obtained by averaging across all persons: $\sum_{\vt}\nicefrac{\phi(\vt,i)}{N}$.
Any \gls{snp} that beats a specified threshold of effect is deemed relevant to Celiac disease by \gls{ours}.
We use a $60-40\%$ train-test split, and outcome model selection is done via cross-validation within the training set.
We did $5$-fold cross-validation using just the training set.
We use Scikit-learn~\cite{pedregosa2011scikit} to fit the outcome models and for cross-validation.

\paragraph{Results}

The best outcome model was a Lasso model, trained with regularization constant $10$.
We select relevant \glspl{snp} by thresholding estimated effects at a magnitude $>0.1$.
From $1050$ \glspl{snp} ($1000$ not reported before) \gls{ours} returned $31$ \glspl{snp}, out of which $13$ were previously reported as being associated with Celiac disease~\citep{dubois2010multiple,sollid2002coeliac,hunt2008novel,adamovic2008association}.
In~\cref{appsec:real-exp} we plot the true positive and false negative rates of identifying previously reported~\glspl{snp}, as a function of the effect threshold.

\begin{wraptable}[14]{r}{0.43\textwidth}
   \centering
   \begin{tabular}{ c c c }
     \toprule
    \gls{snp} & \textsc{Effect.} & \textsc{Coef.} \\ 
    \midrule
 rs13151961 & $0.17$ & $0.32$ \\ 
 rs2237236 & $0.17$ & $0.00$ \\ \midrule
 rs1738074 & $-0.16$ & $-0.23$ \\ 
 rs11221332 & $-0.15$ & $-0.24$ \\ 
    \bottomrule
  \end{tabular}
    \caption{A few \glspl{snp} previously reported as relevant and recovered by \gls{ours}, with
    estimated effects and Lasso coefficients.
    \gls{ours} produces effect estimates that do not rely purely on the coefficients.
      }
   \label{tab:snps-est}
\end{wraptable}
In~\cref{tab:snps-est}, we list a few \glspl{snp} that were both deemed relevant by \gls{ours} and were reported in existing literature~\citep{dubois2010multiple,sollid2002coeliac,hunt2008novel,adamovic2008association}, their effects, and their Lasso coefficients.
The full list is in~\cref{apptab:snps-est} in~\cref{appsec:exps}.
If \gls{ours} cannot adjust for confounding, the Lasso coefficients would dictate the effects; $0$ coefficient means $0$ effect.
However, the two pairs of \glspl{snp} in \cref{tab:snps-est} show that the effects estimated by \gls{ours} do not rely solely on the Lasso coefficients.
For the first pair (rs13151961, rs2237236), the effect is the same but the coefficient of one is $0$, while the other is positive.
We note that rs2237236 was found to be associated with ulcerative colitis~\citep{hindorff2009potential,anderson2011meta}, which
is an inflammatory bowel disease that has been reported to share some common genetic basis with celiac disease~\citep{pascual2014inflammatory}.
For the second pair, (rs1738074, rs11221332), the magnitude of the effect is smaller for the former, but the coefficient is larger.
Thus, \gls{ours} adjusts for confounding factors that the outcome model ignored.
\section{Discussion}

When positivity is violated in \tci, not all effects are estimable without further assumptions.
In such cases, practitioners have to turn to parametric models to estimate causal effects.
However, parametric models can be misspecified when used without underlying causal mechanistic knowledge.
We develop a new general setting of observational causal effect estimation called \glsreset{efc}\gls{efc} where the confounder can be expressed as a function of the data, meaning positivity is violated.
Even when positivity is violated, the effects of many \funcints{} are estimable.
We develop a sufficient condition called \glsreset{func-pos}\gls{func-pos}\ to estimate effects of \funcints. 
Such effects could be of independent interest; like the effect of cumulative dosage of a drug instead of joint effects of multiple dosages at different times.

Second, we prove a necessary condition for nonparametric estimation of effects of the full intervention.
We propose the \gls{c-red} condition, under which, the effect of the full intervention on $\mbt$ is estimable without parametric restrictions.
We develop \glsreset{ours}\gls{ours} that computes surrogate interventions whose effects are estimable and match a conditional effect of interest.
Further, we give bounds on errors (\cref{thm:main-effect-bound}) induced due to imperfect estimation of the surrogate intervention.
Finally, we empirically demonstrate \gls{ours}'s ability to correct for confounding in both simulated and real data.

\paragraph{Future.}  A few directions of improvement remain which we elaborate next.
First, \gls{func-pos} may not hold for all functions $g(\mbt)$ that we want to intervene on. 
Instead, one could compute a ``projection'' $g_\Pi$ to the space of functions that satisfy \gls{func-pos} and inspect the effects defined by $g_\Pi$ instead.
A second direction of interest is to let $h(\mbt)$ only account for a part of the confounding, meaning ignorability is violated.
This bias could be mitigated under smoothness conditions of the outcome function and its interaction with the degree of violation of ignorability.

Finally, \gls{ours}'s search strategy is Euler integration, which is equivalent to gradient descent with a fixed step size.
Optimization techniques like momentum, rescaling the gradient using an adaptive matrix, and using second order hessian information, speed up gradient descent.
However, if there are many local or global minima for $(h(\vt) - h(\vts_2))^2$, such techniques will result in a different solution than Euler integration, which could mean that effect estimates are biased.
One extension of \gls{ours} would allow for search strategies that use such techniques.

\newpage
\section*{Broader Impact}
Our work mainly applies to causal inference where confounders are specified as functions of observed data, such as in problems in genetics and healthcare.
We choose to assess the impact of our work through its applications in these fields.
A positive impact of the work is that better estimates of causal effects helps guide treatment for people and aid in understanding biological pathways of diseases.
However, in healthcare, data collected in hospitals has biases.
If, for instance, a certain demographic of people have more complete data collected about them, then this demographic would have better quality effect estimates, potentially meaning that they receive better treatment.
This problem could be characterized by evaluating the positivity of treatment and completeness of confounders in electronic health record data split by demographics.

\section*{Acknowledgements}

The authors were partly supported by NIH/NHLBI Award R01HL148248, and by NSF Award 1922658 NRT-HDR: FUTURE Foundations, Translation, and Responsibility for Data Science.
The authors would like to thank Xintian Han, Raghav Singhal, Victor Veitch, Fredrik D. Johansson and the reviewers for thoughtful feedback.
The authors would also like to thank Mukund Sudarshan and Prof. Sriram Sankararaman for help with running the \gls{gwas} experiments.

\bibliography{refs}

\newpage
\appendix
\onecolumn
\section{Theoretical details}\label{appsec:theory}
\setcounter{theorem}{0}

\subsection{A note about the assumptions}\label{appsec:assumptions}

\paragraph{Note about the assumptions}
In~\cref{thm:thm-c-red}, assumption 1 consists of three parts that can all be validated on observed data: 1) that the gradient flow converges,
2) that the confounder value of the surrogate matches the confounder value
whose effect is of interest, and
3) that the surrogate intervention lies in the support of the pre-outcome variables.
Assumption 2 is required for expectations and their gradients to exist and be finite.
In~\cref{thm:main-effect-bound}, assumption 1 requires a consistent estimator of $\E[\mby \g \mbt]$, which can be provided with regression.
Assumption 3 lists regularity conditions which help control how the surrogate estimation error propagates to the effect error.

\subsection{Proof of \Cref{thm:thm-c-red}}\label{appsec:thm-c-red}
We restate the theorem for completeness:
\begin{theorem}
  \lemmaone{}
\end{theorem}
\begin{proof}
  Recall definition of conditional effect $\phi(\vt, h(\vt_2)) = \E_\mbeta f(\vt, h(\vt_2), \mbeta)$.
  Recall $\nabla_\vt$ is the gradient \wrt the first argument of $f$, that is $\vt$.
  First, by assumption 2, $\E$ and $\nabla$ commute, under the dominated convergence theorem. Then, by \gls{c-red}
  \[\nabla_{\vt} \phi(\vt, h(\vts))^T \nabla_{\vt} h(\vt) = \nabla_{\vt} \E_{\mbeta} f(\vt,h(\vts),\mbeta)^T \nabla_{\vt} h(\vt) = \E_{\mbeta}  [\nabla_{\vt} f(\vt,h(\vts),\mbeta)^T \nabla_{\vt} h(\vt)] = 0 .\]
  Now consider the gradient flow equation $\nicefrac{d\vt(s)}{ds} = -\nabla_{\vt} (h(\vt) - h(\vts_2))^2$.
  We refer to the gradient evaluated at $\vt$ as $\Delta\vt = - \nabla_{\vt} (h(\vt) - h(\vts_2))^2 = - 2 (h(\vt) - h(\vts_2))\nabla_{\vt} h(\vt)$.
  We will express $\phi(\vtt(\vts, h(\vts_2)), h(\vts_2))$ as defined by the starting point $\phi(\vts, h(\vts_2))$ and the gradient flow equation.

  Let the solution path to the gradient flow equation be $C$ with $\vts, \vtt(\vts, h(\vts_2))$ being the starting and ending points respectively.
  By the Gradient Theorem~\citep{spivak2018calculus}, we have that $\phi(\vts, h(\vts_2))$ and $\phi(\vtt(\vts, h(\vts_2)), h(\vts_2))$ are related via the line integral over $C$:
  \[\int_C \nabla_{\vt} \phi(\vt, h(\vts_2)) \cdot d \vt = \phi(\vtt(\vts, h(\vts_2)), h(\vts_2)) - \phi(\vt, h(\vts_2)) \]
  Let $\vt(s)$ be a parametrization of solution path $C$ by the scalar time $s\in[0,\infty)$.
  Now, to obtain the value of $\phi(\vt, h(\vts_2))$, we will compute the line integral over the vector field defined by $\nabla_{\vt} \phi(\vt, h(\vts_2))$, which exists by assumption 2 in~\cref{thm:thm-c-red}, evaluated along the path $C$ defined by $\Delta \vt(s)$:
  \begin{align}\label{eq:f-function-change}
    \begin{split}
      \phi(\vtt(\vts, h(\vts_2)), h(\vts_2)) &= \phi(\vts, h(\vts_2)) + \int_C \nabla_{\vt} \phi(\vt, h(\vts_2)) \cdot d \vt
      \\
      &= \phi(\vts, h(\vts_2)) + \int_0^\infty \nabla_{\vt} \phi(\vt(s), h(\vts_2))^T \fullder{\vt(s)}{s}\ ds 
      \\
      & =  \phi(\vts, h(\vts_2)) + \int_0^\infty \nabla_{\vt} \phi(\vt(s), h(\vts_2))^T \Delta \vt(s)\ ds 
      \\
      & =  \phi(\vts, h(\vts_2))
      \\ & \quad \quad + \int_0^\infty - 2((h(\vt(s)) - h(\vts_2)))\, \nabla_{\vt} \phi(\vt(s), h(\vts_2))^T \nabla_{\vt} h(\vt(s))\ ds
      \\
      & = \phi(\vts, h(\vts_2)) + 0 \quad \quad \text{\{by \gls{c-red}\}}
    \end{split}
  \end{align}
  Finally, by assumption 1 in~\cref{thm:thm-c-red},  $h(\vtt(\vts,h(\vts_2))) = h(\vts_2)$, and so 
  \begin{align}\label{eq:f-connection}
    \phi(\vts, h(\vts_2)) = \phi(\vtt(\vts, h(\vts_2)), h(\vts_2))
    = \phi(\vtt(\vts, h(\vts_2)), h(\vtt(\vts, h(\vts_2))))
  \end{align}
  For clarity, the same equation, but using $\vtt$ and suppressing dependence on $\vts, h(\vts_2)$):
  
  \begin{align}\label{eq:f-connection}
    \phi(\vts, h(\vts_2)) = \phi(\vtt, h(\vts_2))
    = \phi(\vtt, h(\vtt))
  \end{align}

  Under the causal model for \gls{efc}, the outcome $\mby = f(\mbt, h(\mbt), \mbeta)$.
  Then, $\forall \vt \in \supp(p(\mbt))$,
  \begin{align}
    \label{eq:equality-of-conditional}
  \quad \E[\mby\g \mbt = \vt] = \E_{\mbeta}[f(\vt, h(\vt), \mbeta)] = \phi(\vt, h(\vt)).
  \end{align}
  
  Using that $ \vtt(\vts, \vts_2) \in \supp(p(\mbt))$ and~\cref{eq:equality-of-conditional,eq:f-connection}, the conditional effect is identified
  \begin{align}
    \begin{split}
      \phi(\vts, h(\vts_2)) & = \phi(\vtt(\vts, h(\vts_2)), h(\vtt(\vts, h(\vts_2))))
      \\
      & = \E[\mby\g \mbt=\vtt(\vts, h(\vts_2))]
    \end{split}
  \end{align}
  Thus, the conditional effect, and consequently the \avgeff{}, are identified as $\E[\mby\g \vtt(\vts, h(\vts_2))]$ and ${\tau}(\vts) = \E_{h(\mbt)}\E[\mby\g \vtt(\vts, h(\mbt))]$ respectively.
\end{proof}
\paragraph{Note about convergence of gradient flow}
Any ODE's solution, if it exists and converges, converges to an $\omega$-limit set~\cite{teschl2012ordinary}.
An $\omega$-limit set is nonempty when the solution path lies entirely in a closed and bounded set and can consist of limit cycles, equilibrium points, or neither~\cite{hirsch1974differential,teschl2012ordinary}.
A gradient flow equation $\nicefrac{d \vt (s)}{ ds} = -\nabla h(\vt)$ (also called a gradient system) has the special property that its $\omega$-limit set only consists of critical points of $h(\vt)$; critical points of $h(\vt)$ are also equilibrium points of the gradient flow equation~\cite{hirsch1974differential}.
Further, if $\nabla h(\vt)$ exists and is bounded and $h(\vt)$ has bounded sublevel sets ($\{\vt: h(\vt) \leq c\}$), then the solution to the gradient flow equation will entirely lie within a bounded set.
This is because along the solution path, $h(\vt(s))$ always decreases meaning that the solution will remain in any sublevel set it started in.
Thus, if $h(\vt)$ has bounded sublevel sets, the solution of the gradient flow equation will converge only to critical points of $h(\vt)$.

\subsection{Estimation error in \gls{ours}}\label{appsec:lode-error}

\begin{theorem}
\theoremmain
\end{theorem}
\begin{proof} (of \Cref{thm:main-effect-bound})
  Recall the definition of conditional effect : $\phi(\vt, h(\vt_2)) = \E_{\mbeta}f(\vt, h(\vt_2), \mbeta)$.

  \gls{ours}'s estimate of the conditional effect is $\hat{f}(\vthat(\vts, h(\vts_2)))$.
  We will suppress notation for dependence on $\vts, h(\vts_2)$, and use $\vtt$ and $\vthat$ to refer to the true surrogate intervention and the estimated surrogate interventions respectively.
  Note $\hat{f}$ is the estimate of the conditional expectation $\E[\mby\g \mbt=\vt]$, learned from $N$ samples.
  We first bound the error by splitting into two parts and bounding each separately:
  \begin{align*}
      |\xi(\vts,h(\vts_2))| & =  | \hat{f}(\vthat)  -  \phi(\vts, h(\vts_2)) |
      \\
                            &\leq| \hat{f}(\vthat)  -  \phi(\vthat, h(\vthat)) | + | \phi(\vthat, h(\vthat))  - \phi(\vts, h(\vts_2))| 
      \\
                          & \leq c(N) + | \phi(\vthat, h(\vthat))  - \phi(\vts, h(\vts_2))| 
      \\
                          & \leq |\phi(\vthat, h(\vthat)) - \phi(\vthat, h(\vts_2)) | + |\phi(\vthat, h(\vts_2))  - \phi(\vts, h(\vts_2)) | + c(N)
  \end{align*}
  The first term is bounded via the Lipschitz-ness of $\phi$ as a function of $h(\vt)$ with fixed first argument $\vt = \vthat$.
  \[|\phi(\vthat, h(\vthat)) - \phi(\vthat, h(\vts_2)) | \leq L_{z,\vthat} |h(\vthat) - h(\vts_2)|\]

  We now bound the remaining term.
  Recall that \Gls{ours}'s computation of the surrogate intervention involved $K$ gradient steps, each of size $\ell$.
  We work with a constant step-size but the analysis can be generalized to a non-uniform step size.
  Indexing steps with $i$, let $d_i = h(\vt_i) - h(\vts_2)$ be the confounder mismatch error at the $i$th iterate.
  Then note that $\vthat = \vts - \ell \sum_{i=0}^{K-1} 2 d_i \nabla_\vt h(\vt_i) $.
  We can use this to bound the error $\phi(\vthat, h(\vts_2)) - \phi(\vts, h(\vts_2)) $.
  With $\vt_K = \vthat$ and $\vt_0 = \vts$, we proceed by expressing the error as a telescoping sum and using the Taylor expansion for $\phi(\vt, h(\vts_2))$ in terms of the the first argument $\vt$.
  \allowdisplaybreaks
  \begin{align}\label{eq:lode-est-err}
      &\phi(\vthat, h(\vts_2))  - \phi(\vts, h(\vts_2))  = \sum_{i=0}^{K-1} \phi(\vt_{i+1}, h(\vts_2))  - \phi(\vt_{i}, h(\vts_2))
      \\
                      & = \sum_{i=0}^{K-1} \nabla_\vt \phi(\vt_i, h(\vts_2))^\top (\vt_{i+1} - \vt_i)  
                      \\ & \quad\quad \quad \quad + 
                      \frac{1}{2}(\vt_{i+1} - \vt_i)^\top \nabla_\vt^2  \phi(\vt_{i}, h(\vts_2)) (\vt_{i+1} - \vt_i) + \mathcal{O}(\|\vt_{i+1} - \vt_i\|_2^3)
                      \\
                      & = \sum_{i=0}^{K-1} 2 \ell d_i \nabla_\vt \phi(\vt_i, h(\vts_2))^\top \nabla_\vt h(\vt_i)
                        + 2(\ell d_i)^2 \nabla_\vt h(\vt_i)^\top \nabla_\vt^2  \phi(\vt_{i}, h(\vts_2)) \nabla_\vt h(\vt_i)+ \mathcal{O}(\ell^3)
                      \\
                      & = \sum_{i=0}^{K-1} 0  + 2(\ell d_i)^2 \nabla_\vt h(\vt_i)^\top \nabla_\vt^2  \phi(\vt_{i}, h(\vts_2)) \nabla_\vt h(\vt_i) + \mathcal{O}(\ell^3)
                      \\
                      & = \mathcal{O}(K\ell^3) + \sum_{i=0}^{K-1} 2(\ell d_i)^2 \nabla_\vt h(\vt_i)^\top \nabla_\vt^2  \phi(\vt_{i}, h(\vts_2)) \nabla_\vt h(\vt_i) 
                      \\
                      & \leq \mathcal{O}(K\ell^3) + \sum_{i=0}^{K-1} 2(\ell (h(\vt_i) - h(\vts_2)))^2 \left|\nabla_\vt h(\vt_i)^\top \nabla_\vt^2  \phi(\vt_{i}, h(\vts_2)) \nabla_\vt h(\vt_i) \right|
                      \\
                      & \leq \mathcal{O}(K\ell^3) + \sum_{i=0}^{K-1}2\ell^2M\left| \nabla_\vt h(\vt_i)^\top \nabla_\vt^2  \phi(\vt_{i}, h(\vts_2)) \nabla_\vt h(\vt_i) \right|
                      \\
                      & \leq \mathcal{O}(K\ell^3) + \sum_{i=0}^{K-1} 2\ell^2M \sigma_{\mathtt{H} \phi } \|\nabla_\vt h(\vt_i)\|_2^2
                      \\
                      & \leq \mathcal{O}(K\ell^3) + \sum_{i=0}^{K-1} 2\ell^2M \sigma_{\mathtt{H} \phi } L^2_{h} 
                      \\
                      & = 2K \ell^2 \left(\mathcal{O}(\ell) + M \sigma_{\mathtt{H} \phi } L^2_{h}  \right),
  \end{align}
  where the inequalities follow by the maximum value of $(h(\vt_i) - h(\vts_2))^2$, bounded eigenvalues of the Hessian of $\phi$ and the Lipschitz-ness of $h(\vt)$.
  
  Another way we bound the error is via the Lipschitz constant of the conditional expectation as a function of $\vt$.
  Recall this is $L_{e}$. An alternate bound on the error is as follows:
\begin{align*}
|\phi(\vthat, h(\vthat))  - \phi(\vts, h(\vts_2))| = |\phi(\vthat, h(\vthat))  - \phi(\vtt, h( \vtt))| \leq & L_e\|\vtt - \vthat\|_2
\end{align*}
The bound follows:
\begin{align*}
  |\xi(\vt,h(\vts_2))|  & \leq  c(N) + \min\left(L_e\|\vtt - \vthat\|_2, 
                        \quad 2K \ell^2 \left(\mathcal{O}(\ell) + M \sigma_{\mathtt{H} \phi } L^2_{h}  \right) + 
                      L_{z,\vthat} \|h(\vthat) - h(\vts_2)\|_2\right)
\end{align*}
\end{proof}
\subsubsection{A note on linear confounder functions and \gls{ours}}\label{appsec:linear-conf-func}
In the proof above, the error in Euler integration accumulates due to terms like this one: $\nabla_\vt^\top h(\vt) \nabla^2_\vt f(\vt, h(\vts),\eta) \nabla_\vt h(\vt)$.
For a linear confounder function that satisfies $\nabla_\vt h(\vt) = \beta$, such terms can be expressed as $\beta^\top \nabla_\vt (\nabla_\vt f(\vt, h(\vts),\eta)^\top \beta) = \beta^\top \nabla_\vt (0)= 0 $ under \gls{c-red}.
Thus, such error does not accumulate even with large step sizes.

Further, note that the gradient flow equation in \gls{ours} for the causal model $A$ in~\cref{sec:exps} is a linear ODE whose solution has a closed form expression and one can estimate the surrogate without numerical integration~\cite{teschl2012ordinary}.

\subsection{Proof of sufficiency of Effect Connectivity}\label{appsubsec:sufficiency}

\begin{theorem}
  \sufficiencylemma{}{}
\end{theorem}
\begin{proof}
  Recall  $\phi(\vt, h(\vt)) = \E_\mbeta f(\vt,h(\vt), \mbeta)$.
  We have $\forall \vts \in\supp(p(\mbt))$:
  \[ p(h(\mbt) = h(\vts_2))>0\implies p(\phi(\mbt, h(\mbt)) = \phi(\vts, h(\vts_2))\g h(\mbt) = h(\vts_2)) >0.\]
  This implies $\exists \vtt\in \supp(\mbt), \phi(\vtt, h(\vts_2)) = \phi(\vts, h(\vts_2)), \quad s.t. \quad h(\vtt) = h(\vts_2).$
  
  Then, $\phi(\vts, h(\vts_2)) = \phi(\vtt, h(\vts_2)) = \phi(\vtt, h(\vtt)) = \E[\mby\g \mbt=\vtt].$
\end{proof}

\subsection{Necessity of Effect Connectivity for Nonparametric effect estimation in \gls{efc}}\label{appsubsec:necessity-proof}

\begin{theorem}
  \necessitylemma{}
\end{theorem}
\begin{proof} (Proof of~\Cref{thm:thm-necessity})
  Let the outcome be $\mby= f(\mbt, h(\mbt))$.
  Recall the joint distribution $p(\mbt, \mby)$ and let $h(\mbt)$ be the confounder.
  Let Effect Connectivity be violated, i.e.\ there exists a non-measure-zero subset $B \in \supp(\mbt)\times \supp(h(\mbt))$  such that \footnote{Non-zero w.r.t.\ the product measure over $\supp(\mbt)\times \supp(h(\mbt))$ due to $p$.}:
  \[\forall \,\, \vt, h(\vt_2) \in B, \quad \quad p(f(\mbt, h(\mbt)) = f(\vt, h(\vt_2))\g h(\mbt) = h(\vt_2)) = 0.\]
  Now, we construct a new outcome $\mby_2 = f_2(\mbt, h(\mbt))$ and show the conditional effects for this new outcome are different from the one defined by $f$ on $\forall (\vt, h(\vt_2))\in B$. Let
  \[f_2(\vt, h(\vt_2)) = f(\vt, h(\vt_2)) + 10*1((\vt, h(\vt_2))\in B)|.\]
  We have $f_2(\vt, h(\vt)) = f(\vt, h(\vt)) \, \forall \vt \in \supp(\mbt)$ , as the additional term in $f_2$ is only present for $(\vt, h(\vt_2))\in B$; this follows from the fact that $\forall \vt \in \supp(\mbt)$, $(\vt, h(\vt)) \not\in B$ as
  \[p[f(\mbt, h(\mbt)) = f(\vt, h(\vt)) \g h(\mbt) = h(\vt)] = p[f(\mbt, h(\mbt)) = f(\vt, h(\vt))]  >0.\]
  Thus, $p(\mby, \mbt) =^d p(\mby_2, \mbt)$ are equal in distribution since $B \cap \supp(\mbt, h(\mbt)) = \emptyset$.
  This means that the conditional effects are different for the outcomes $\mby, \mby_2$ for all $(\vt, h(\vt_2)) \in B$: 
  \[\E[\mby \g do(\mbt = \vt), h(\mbt) = h(\vt_2)] \not= \E[\mby_2 \g do(\mbt = \vt), h(\mbt) = h(\vt_2)] \]
  Therefore, for causal models that violates Effect Connectivity, there exist observationally equivalent causal models with different causal effects.
  Thus, nonparametric effect estimation is impossible. 
  Thus, Effect Connectivity is required for \gls{efc}.
\end{proof}

\subsection{Algorithmic details}
We give in ~\cref{alg:full-eff-estimation} pseudocode for \gls{ours}.
  \begin{algorithm}[ht]
    \SetAlgoLined
    \SetKwProg{Fn}{Function}{ is}{end}
    \KwIn{Functional confounder $h(\mbt)$; tolerance $\epsilon$}
    \KwOut{Conditional effects of $\vts, h(\vts_2)$}
    Regress $\mby$ on $\mbt$ and compute 
    $\hat{f}() := \arg\min_{u\in \mathcal{F}} \E_{\mby,\mbt}(\mby - u(\mbt))^2 .$ \\
    To estimate effects of $\vts, h(\vts_2)$, compute the surrogate intervention $\vtt(\vts, h(\vts_2))$ by Euler integrating the gradient flow equation, initialized at $\vt = \vts$, until $(h(\vt_s) - h(\vts_2))^2 < \epsilon$.
    \[\fullder{\vt(s)}{s} = \nabla_\vt (h(\vt_s) - h(\vts_2))^2,\]
    \\
    Return $\hat{f}(\vtt(\vts, h(\vts_2)))$;
    \caption{\gls{ours} for $do(\mbt=\vts)$}
    \label{alg:full-eff-estimation}
  \end{algorithm}
  \paragraph{Extensions of \gls{ours}}
  Consider that we have access to $m(h(\mbt))$ for some bijective differentiable function $m(\cdot)$, instead of $h(\mbt)$.
  The orthogonality in \gls{c-red} holds $\nabla_{\vt} f(\vt,h(\vt_2), \eta)^T \nabla_{\vt} m(h(\vt)) = m'(h(\vt))\nabla_{\vt} f(\vt,h(\vt_2), \eta)^T \nabla_{\vt} h(\vt) = 0$.
  Then, using $m(h(\vt))$ to compute the surrogate $\vtt(\vts, h(\vts_2))$, \gls{ours} would estimate valid effects. 
  Similarly, \gls{ours} can estimate the effect on any differentiable transformation of the outcome $m(y)$, because
  $\nabla_{\vt} m(y_\vt)^T\nabla_{\vt} h(\vt) = m'(y_\vt)\nabla_{\vt} f(\vt,h(\vt_2), \eta)^T \nabla_{\vt} h(\vt) = 0$ holds.
\section{Experimental Details}\label{appsec:exps}
\subsection{Functional confounders in \gls{gwas}}
Here, we show how $h(\mbt)= At$ and $A$ reflect the traditional \gls{pca} based adjustment in \gls{gwas}.
Recall population structure acts as a confounder in \gls{gwas}.
\citet{price2006principal} demonstrated that using the principal components of the normalized genetic relationships matrix adjusts for confounding due to population structure in \gls{gwas}.
Let the genotype matrix be $G$ with people as rows and \snps{} as columns, such that each element is one of $0,\nicefrac{1}{2}, 1$, where $\nicefrac{1}{2}$ and $1$ refer to one and two copies of the allele respectively at the position of the \gls{snp}.
With $p_s$ as the allele frequency at \gls{snp} $s$~\citep{thorntonsummer}, $\Phi$ is the genetic relationship matrix whose elements are defined as 
$\Phi_{i,j} =\frac{1}{S} \sum_{s=1}^S \nicefrac{(G_{i,s} - p_s)(G_{j,s} - p_s)}{p_s(1 - p_s)}$.
Then,~\citet{price2006principal} compute the top $K$ ($10$ suggested) principal components of $\Phi$ to use as the axes of variation due to the population structure.
The eigenvectors of $\Phi$ are the left eigenvectors of $\hat{G}$ such that $\Phi = \hat{G}\hat{G}^T$ which capture independent axes of variation of individuals.

\citet{price2006principal} exploit the idea that if a \gls{snp} aligns with some of the axes of variation, this is due to the population structure.
These axes of variation are the top $K$ eigenvectors $U$ of $\phi = \hat{G}\hat{G}^T \approx U\Lambda U^\top$, where $U\in \mathbb{R}^{N\times K}$, $\Phi\in\mathbb{R}^{N\times N}$ and $\Lambda \in \mathbb{R}^{K\times K}$.
Here, $U$ are also the left singular vectors of $\hat{G} \approx U\Sigma V^T$ where $\Sigma\in \mathbb{R}^{K\times K}$ is diagonal, and $V\in \mathbb{R}^{S\times K}$.
We use $\approx$ to denote that the chosen $K$ eigenvectors explain the variation due to population structure; what remains are random mutations.

Let the $s$th \gls{snp} be $\hat{G}_{\cdot, s}\in \mathbb{R}^{N}$, which is a column in $\hat{G}$.
In \citet{price2006principal}, population structure in the $s$th \gls{snp} is captured in $\hat{G}_{\cdot, s}^\top U$.
In words, projecting the \gls{snp} $\hat{G}_{\cdot, s}$ onto the axes of variation in individuals gives the population structure between $s$th \gls{snp} and the outcome.
This projection $\hat{G}_{\cdot, s}^\top U$ is a row of $\hat{G}^\top U\in \mathbb{R}^{S\times K}$.
In turn, $\hat{G}^\top U\in \mathbb{R}^{S\times K}$ is the population structure in all \glspl{snp}.
Projecting this population structure onto the genotype of an individual gives the confounding due to population structure amongst the \glspl{snp} present in the genotype.
With $G_{j, \cdot}\in \{0,\nicefrac{1}{2}, 1\}^{S}$ as the genotype for an individual $j$, this projection is $\left((\hat{G}^\top U)^\top G_{j, \cdot}\right)$.
However, $\hat{G} \approx U\Sigma V^T$ implies that $\hat{G}^\top U\approx V\Sigma$.
Reflecting this, $h(\mbt) = \Sigma V^T \mbt$ is the functional confounder for an individual $\mbt$.

\newpage
\subsection{Expanded results}\label{appsec:real-exp}
In~\cref{apptab:snps-est}, we list the $13$ \glspl{snp} recovered by \gls{ours}, that have been previously reported as relevant to Celiac disease. In~\cref{fig:tp-fn-rate}, we plot the true positive and false negative rate amongst \glspl{snp} deemed relevant by \gls{ours}.
The ground truth here are the \glspl{snp} reported associated with celiac disease in prior literature. \hspace*{\fill}
\begin{wraptable}{r}{0.45\textwidth}
   \centering
   \begin{tabular}{ c c c }
     \toprule
    \gls{snp} & \textsc{Effect} & \textsc{Lasso Coef.} \\ 
    \midrule
    rs3748816 & $0.12$ & $0.20$ \\ 
 rs10903122 & $0.10$ & $0.17$ \\ 
 rs2816316 & $0.11$ & $0.20$ \\ 
 rs13151961 & $0.17$ & $0.32$ \\ 
 rs2237236 & $0.17$ & $0.00$ \\ 
 rs12928822 & $0.14$ & $0.29$ \\
 rs2187668 & $-0.70$ & $-2.37$ \\ 
 rs2327832 & $-0.12$ & $-0.20$ \\ 
 rs1738074 & $-0.16$ & $-0.23$ \\ 
 rs11221332 & $-0.15$ & $-0.24$ \\ 
 rs653178 & $-0.13$ & $-0.21$ \\ 
 rs4899260 & $-0.12$ & $-0.19$ \\ 
 rs17810546 & $-0.12$ & $-0.20$ \\ 
    \bottomrule \\
  \end{tabular}
    \caption{Full list of \glspl{snp} previously reported as relevant that were recovered by \gls{ours}, and their
    estimated effects and Lasso coefficients for \glspl{snp}. The effect threshold here is $0.1$.}
   \label{apptab:snps-est}
\end{wraptable}
\begin{wrapfigure}[15]{r}{0.45\textwidth}
\centering
\includegraphics[width=0.5\textwidth]{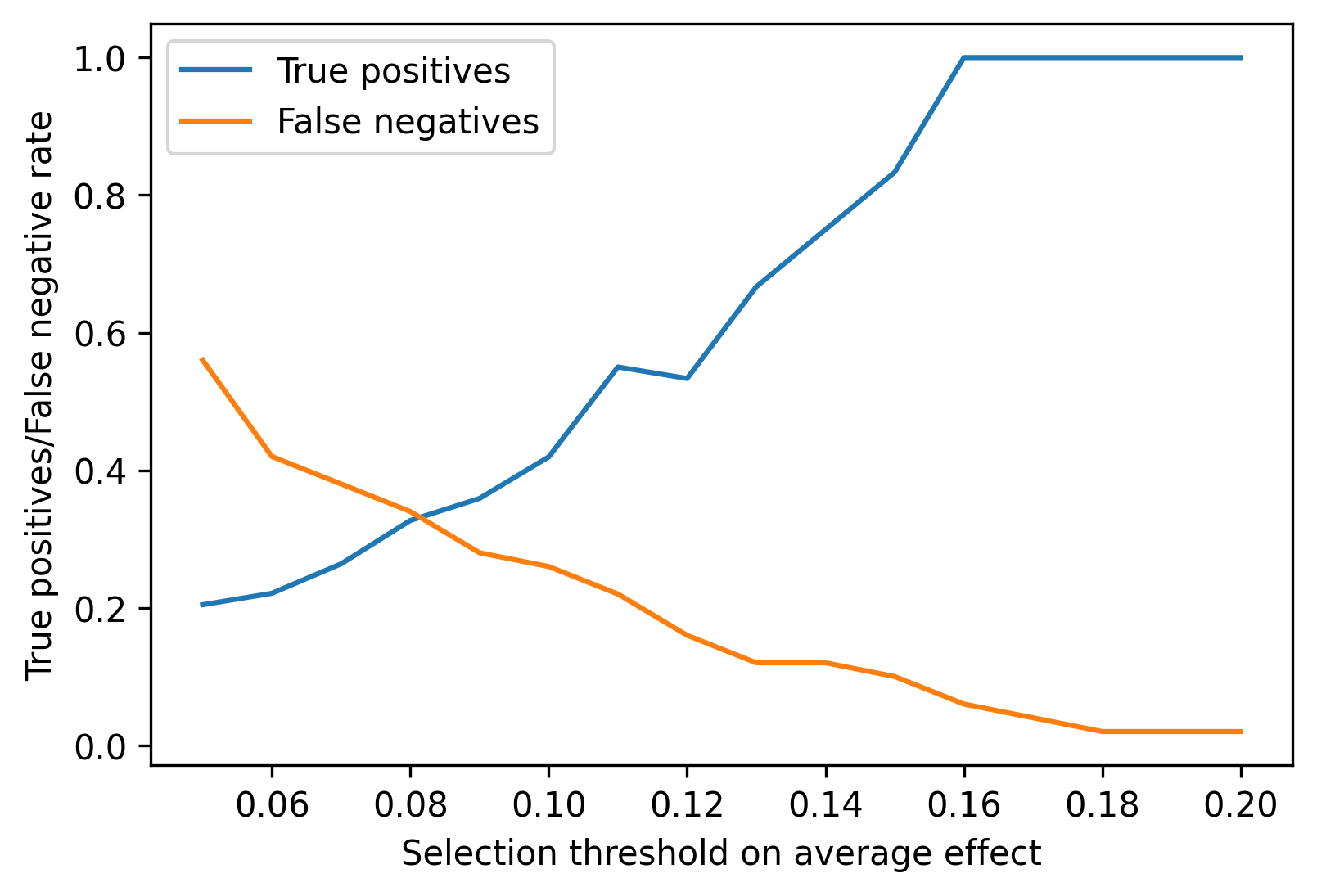}
\caption{
  True positive vs. False negative rate as we vary the threshold on average effects, that determines which \snps{} \gls{ours} deems relevant to the outcome.
}
\label{fig:tp-fn-rate}
\end{wrapfigure}

\end{document}